\documentclass[]{interact}

\usepackage{epstopdf}
\usepackage{subfigure}

\usepackage{natbib}
\bibpunct[, ]{(}{)}{;}{a}{}{,}
\renewcommand\bibfont{\fontsize{10}{12}\selectfont}

\usepackage{amssymb,amsmath,mathrsfs,stmaryrd,amsthm,mathtools,amsfonts,graphicx,hyperref,url,booktabs,nicefrac,comment,microtype,tikz-cd,lingmacros,nccmath,tree-dvips, bbm, bm, setspace, etoolbox, algorithm, color, enumitem,lineno,physics, relsize} 

\DeclareMathOperator*{\argmax}{argmax}

\theoremstyle{plain}
\newtheorem{theorem}{Theorem}[section]

\newtheorem{corollary}[theorem]{Corollary}
\newtheorem{prop}[theorem]{Proposition}
\theoremstyle{definition}
\newtheorem{definition}[theorem]{Definition}

\theoremstyle{remark}

\begin{document}


\title{Structural clustering of volatility regimes for dynamic trading strategies}

\author{
\name{Arjun Prakash\textsuperscript{a}$^\ast$, Nick James\textsuperscript{b}$^\ast$, Max Menzies\textsuperscript{c}$^\ast$\thanks{Contact Max Menzies. Email: max.menzies@alumni.harvard.edu} \thanks{$^\ast$Equal contribution} and Gilad Francis\textsuperscript{d}}
\affil{\textsuperscript{a}School of Computer Science, University of Sydney, NSW, Australia; \textsuperscript{b}School of Mathematics and Statistics, University of Melbourne, VIC, Australia; \textsuperscript{c}Yau Mathematical Sciences Center, Tsinghua University, Beijing, China; \textsuperscript{d}School of Mathematics and Statistics, University of Sydney, NSW, Australia. }
}

\maketitle

\begin{abstract}
We develop a new method to find the number of volatility regimes in a nonstationary financial time series by applying unsupervised learning to its volatility structure. We use change point detection to partition a time series into locally stationary segments and then compute a distance matrix between segment distributions. The segments are clustered into a learned number of discrete volatility regimes via an optimization routine. Using this framework, we determine a volatility clustering structure for financial indices, large-cap equities, exchange-traded funds and currency pairs. Our method overcomes the rigid assumptions necessary to implement many parametric regime-switching models, while effectively distilling a time series into several characteristic behaviours. Our results provide significant simplification of these time series and a strong descriptive analysis of prior behaviours of volatility. Finally, we create and validate a dynamic trading strategy that learns the optimal match between the current distribution of a time series and its past regimes, thereby making online risk-avoidance decisions in the present.
\end{abstract}

\begin{keywords}
Volatility modelling; regime-switching model; change point detection; spectral clustering; trading strategies
\end{keywords}

\section{Introduction}
\label{intro}

Modelling the volatility of a financial asset is an important task for traders and economists. Volatility may be modelled from an individual stock level to an index level; the latter can represent the uncertainty or systemic risk of an entire sector or economy. Financial markets are both significant in their own right and have substantial flow-on effects on the rest of society, as seen during the global financial crisis, US-China trade war, and COVID-19 pandemic.

Statistical methods for volatility modelling have long been popular in the literature \citep{Shah2019, Kirchler2007, Baillie2009}. Long-standing parametric methods such as ARCH \citep{Engle1982,Zakoian1994} and GARCH  \citep{Bollerslev1986} model the volatility of individual stocks, obeying assumptions such as \emph{stylized facts} \citep{Guillaume1997}, and appropriately choosing parameters to best fit past data. These methods allow traders to model future returns, provided that their founding assumptions continue to hold. Volatility may also be studied indirectly through options pricing \citep{Gonalves2006,Cont2002,Tompkins2001}.

For many years, statisticians have noted that many real-world processes may exhibit significant nonstationarity. These switches in behaviour may disrupt the use of parametric or nonparametric models to predict future behaviour and so must be accounted for. For this purpose, a broad literature has been developed that aims to study nonstationary time series by partitioning them into locally stationary segments. Early research in a purely statistical context was carried out by \citep{Priestley1965,Priestley1973,OzakiTong,Tong1980} with significant advances by \citet{Dahlhaus1997}.

Such concepts have been applied to financial time series to develop \emph{regime-switching models}. These combine the aforementioned statistical techniques with the observation that financial assets exhibit such switching patterns, moving between periods of heightened volatility and ease \citep{Hamilton1989, Lavielle, Lamoureux1990}. Numerous regime-switching models have been developed to model such patterns \citep{Hamilton1989,Klaassen2002, Guidolin2007, Yang2018JEDC, deZeeuw2012}; these are also generally parametric, building on ARCH and GARCH, and must \textit{a priori} specify the number of regimes and underlying distributions. Typically, the number of regimes is assumed to be two \citep{Taylor1999,Taylor2005}, one for high volatility and one for low volatility \citep{Guidolin2011}. However, market crises such as the global financial crisis or COVID-19 may require more flexibility in volatility modelling, through varied assumptions regarding both individual distributions and the number of regimes \citep{Arouri2016, Song2016,Balcombe2017, Carstensen2020, CerboniBaiardi2020, Campani2021}. For example, \citet{Baba2011} describes three regimes in the VIX index from 1990-2010: tranquil, turmoil and crisis. Most recent regime-switching models have been statistical in nature, for instance building on the Heston model \citep{Goutte2017,Papanicolaou2013,Song2018}.

This paper introduces a new method from an alternative perspective to analyze the regime-switching structure of a financial time series and use this for risk-avoidance in the future. Our method flexibly determines the number of volatility regimes, making no assumptions about the underlying data generating process. We use the entire history of the time series to learn its fundamental structure over time and how it behaves in volatile and non-volatile periods alike, improving on most parametric methods that only use recent behaviour. By not assuming the number of regimes \textit{a priori}, we ameliorate a typical shortcoming of such models \citep{Ang2012}. Moreover, our determination of the number of volatility regimes fits well within the existing literature as one could conceivably take our determined number for a given time series and input this learned number in an alternative regime-switching model.

Subsequently, we draw on our findings and use additional learning procedures to design a dynamic trading strategy that can identify volatile periods and allocate capital in real time. We learn the volatility structure of the S\%P 500 and determine whether the present period should be avoided by comparing its empirical distribution with past segments and determining the closest match by distance minimization. We show that it provides superior risk-adjusted returns to the S\&P 500 index in various market conditions and suitably avoids the index during periods of higher volatility, generally associated with market downturns. Our procedure is inspired by both machine learning and metric geometry and contains additional optimization relative to previous trading strategies \citep{Nystrup2016}.

In Section \ref{methodology}, we describe our methodology in detail, together with mathematical proof of its efficacy under specific circumstances. In Section \ref{results}, we validate our methodology on synthetic data, comparing several variants, and show a reasonable clustering structure can be determined for major indices, stocks, popular exchange-traded funds (ETFs) and currency pairs. In Section \ref{trading strat}, we introduce our dynamic trading strategy, incorporating our insights and additional learning procedures. Section \ref{conclusion} summarizes our findings and describes future research and applications of our methods.

\section{Mathematical model}
\label{methodology}

In mathematical statistics, a \emph{time series} $(X_t)$ is a sequence of random variables - measurable functions from a probability space to the real numbers - indexed by time. In finance, one generally conflates the random variable with the observed data point at each time. As such, a \emph{financial time series} is a sequence of price data. In this paper, we will examine the time series of closing prices (across a range of assets) $(p_t)_{t\geq 0}$ at time $t$, and the associated log returns $(R_t)_{t\geq 1}=\log ( \frac{p_{t}}{p_{t-1}}) $.

In this section, we describe our method in detail. We begin by assuming our nonstationary time series are generated from Dahlhaus locally stationary processes \citep{Dahlhaus1997} and proceed to partition the time series into stationary segments; specifically, we detect changes in the volatility of a time series via the Mood test change point method. Next, using the empirical cumulative density function of each segment, we use the Wasserstein metric to quantify distance between these distributions. We determine an allocation into an appropriate number of clusters via self-tuning spectral clustering \citep{zelnik2004self}. Thus, we classify our segments of volatility into discrete classes without assuming the number of classes \textit{a priori}. Finally, we record the number of clusters and their structure.

Our methodology differs from usual regime-switching models in that it is inspired by unsupervised learning rather than stochastic volatility models. We use spectral clustering rather than simply testing for low or high volatility in the different segments precisely because two volatility regimes are not suitable in all contexts \citep{Campani2021}. As we study long histories of these time series, rather than recent behaviour for use in probabilistic models, we must account for the possibility of a varying number of regimes between different assets. Periods of ordinary market behaviour, turmoil, and crises, could potentially manifest themselves in three regimes for equities \citep{Baba2011}, but perhaps the behaviour of currency pairs in these times may be quite different.

The precise method that we describe below, applicable to volatility clustering, is not exhaustive. As long as there is consistency between the regime characteristic of interest, the change point algorithm (and its test statistic if applicable), and the distance metric between distributions, the method below could easily be reworked for detection and classification of regimes of alternative characteristics. For example, we could substitute the Jensen–Shannon divergence for the Wasserstein metric, substitute the Bartlett test for the Mood test, or smooth out the empirical distributions via kernel density estimation. Below, we outline the detailed implementation to model volatility and discuss some of the aforementioned alternatives in the process.

\subsection{Partition of the time series}
\label{partition}

Given time series price data, begin by forming the log return time series $(R_t)_{t=1,...,T}$ over a particular time interval. It is generally appropriate to assume the log returns independent random variables with mean close to 0, but not appropriate to assume they have any particular distribution \citep{Ross2013}.

With this in mind, we apply the nonparametric \emph{Mood test}, performed in the CPM package of \citet{Ross_cpm}, to detect changes in the volatility of a time series. Although this is commonly known as a median test, it is also appropriate for detecting change in the variance between two distributions, as described in Section 4 of \citet{mood1954}. More details on the change point framework and our specific implementation can be found in Appendix \ref{cpa appendix}. This yields a collection of change points $\tau_1,...,\tau_{m-1}$. For notational convenience, set $\tau_0=1, \tau_{m}=T$. The stationary segments according to this partition are then
$$(R_t)_{t \in [\tau_{j-1}, \tau_{j}]}, j = 1,2,...,m. $$
This yields $m$ stationary segments. Now let $(Y^{(j)})$ be the restricted time series whose entries are taken from the time interval $[\tau_{j-1}, \tau_{j}]$. That is, $(Y^{(j)}_t)$ consists of the values $R_t$ where $t$ ranges from $\tau_{j-1}$ to $\tau_{j}$. Each $(Y^{(j)})$ has been determined by the algorithm to be sampled from a consistent distribution.

\subsection{Computation of the Wasserstein distance}
\label{Wasserstein}

Next, we compute the Wasserstein distance between the empirical distributions of each stationary segment $(Y^{(j)})$. The Wasserstein metric \citep{Kolouri2017}, also known as the earth mover's distance, is the minimal work to move the mass of one probability distribution into another. Given probability measures $\mu,\nu$ on Euclidean space $\mathbb{R}^d$, we define
\begin{equation*}
    W_{p} (\mu,\nu) = \inf_{\gamma} \bigg( \int_{(x,y)\in \mathbb{R}^d \times \mathbb{R}^d} ||x-y||^p d\gamma(x,y)  \bigg)^{\frac{1}{p}}.
\end{equation*}
This infimum is taken over all joint probability measures $\gamma$ on $\mathbb{R}^d\times \mathbb{R}^d$ with marginal probability measures $\mu$ and $\nu$. In the case where $d=1$, this distance can be computed relatively simply in terms of the cumulative distribution functions associated to the two distributions. Given probability measures $\mu,\nu$ with cumulative distribution functions $F,G$, the distance $W_{p} (\mu,\nu)$ may be computed \citep{DelBarrio} as 
\begin{align*}
    \left(\int_{\mathbb{R}} |F - G|^p dx\right)^\frac{1}{p}.
\end{align*}
This allows us to form a $m \times m$ distance matrix of Wasserstein distances
$$D_{ij}=W_p(\mu_i, \mu_j)=W_p(F_i,F_j), i,j=1,...,m$$
between the $m$ distributions of each locally stationary segment of the log return time series. In our implementation, we set $p=1$.

\subsection{Implementation of spectral clustering}
\label{spectral clustering}

In this section, we cluster the distributions by applying spectral clustering directly to the matrix $D$. This is a natural choice - alternative methods such as K-means would require the data to lie in Euclidean space. Spectral clustering often proceeds with the number of clusters $k$ chosen \textit{a priori} - we explore two methods for making this selection.

We proceed according to the self-tuning spectral clustering methodology proposed by \citet{zelnik2004self}. From the distance matrix $D$, an affinity matrix $A$ is defined according to 
$$A_{i,j} = \exp\left(\frac{-D_{ij}^2}{\sigma_i \sigma_j}\right), $$
where $\sigma_i$ are parameters to be chosen. In their original implementation, \citet{zelnik2004self} select $\sigma_i=D_{i,7}$ to be the $K=7$th neighbour of the point $i$. To adapt this method to our needs, since occasionally there are fewer than 7 segments, we follow \citet{knn} and set $K=\lceil\sqrt{m}\rceil$, where $m$ is the number of segments and set $\sigma_i=D_{i,K}$.

Next, one forms the \emph{Laplacian} $L$ and \emph{normalized Laplacian} $L_{\text{sym}}$ following \citet{Luxburg2007}. We define the diagonal degree matrix $\text{Deg}_{ii}= \sum_j A_{ij}$, and then 
\begin{align*}
    L&=\text{Deg}- A;\\
    L_{\text{sym}}&=\text{Deg}^{-1/2}L\text{Deg}^{-1/2}.
\end{align*}
$L,L_\text{sym}$ are $m \times m$ symmetric matrices, and hence are diagonalizable with all real eigenvalues. By the definition of $L$ and the normalization $L_\text{sym}$, all their eigenvalues are non-negative, $0=\lambda_1 \leq ... \leq \lambda_m$.

We now describe two methods of determining the number of clusters $k$, which we will empirically investigate in Section \ref{results}. The first method is that described by \citet{zelnik2004self}. Obtaining the number of clusters proceeds by varying the number of eigenvectors $c \in [1,m]$:
\begin{enumerate}
  \item Find $x_1, ..., x_c$ the eigenvectors of $L$ corresponding to the $c$ largest eigenvalues. Form matrix $X = [x_1, ..., x_c] \in \mathbb{R}^{n \times c}$. 
  \item Recover the rotation $\mathcal{R}$ which best aligns $X$’s columns with the canonical coordinate system using the incremental gradient descent scheme. 
  \item Grade the cost of the alignment for each cluster number, $c \in [1,m]$, according to \cite[p. 5, Eq 3]{zelnik2004self}. 
\end{enumerate}
The final number of clusters $k^{ZP}$ is chosen to be the value of $c$ with the minimal alignment cost. The second method chooses $k$ to maximize the eigengap between successive eigenvalues. We denote this as $k^e$, defined by $$k^e= \argmax_{c=1,...,m} \lambda_c - \lambda_{c-1}.$$

Having determined the number of clusters $k$, spectral clustering proceeds as follows. We compute the normalized eigenvectors $u_1,...,u_k$ corresponding to the $k$ smallest eigenvalues of $L_\text{sym}$. We form the matrix $U \in \mathbb{R}^{m \times k}$ whose columns are $u_1,...,u_k$. Let $v_i \in \mathbb{R}^k$ be the rows of $U$, $i=1,...,m$. These rows are clustered into clusters $C_1,...,C_k$ according to $k$-means. Finally, we output clusters $A_l=\{i: v_i \in C_l \}, l=1,...,k$ to assign the original $m$ elements, in this case segment distributions, into the corresponding clusters.

\subsection{Theoretical justification of model}
\label{theoretical justification}
In this section, we introduce a mathematical framework: under certain conditions (introduced in Definition \ref{spread defn}), our aforementioned methodology is proven to accurately cluster segments by high or low volatility. Propositions \ref{alt defn prop} and \ref{technicallemma} are intermediate technical results that allow us to establish this claim in Theorem \ref{main thm}.

For the remainder of this section, let $f,g,f_i$ denote continuous probability density functions (pdf's) and let $F,G,F_i$ denote their respective cumulative density functions (cdf's). We recall the necessary properties. Such a function $f$ is continuous $f: \mathbb{R} \to \mathbb{R}$, non-negative everywhere, with $\int_{-\infty}^\infty f(t) dt =1$. The cumulative density function $F$ is defined by $F(x)=\int_{-\infty}^x f(t) dt$.

\begin{definition}
Say a continuous probability density function $f$ is \emph{median-zero} if one of three equivalent conditions hold:
\begin{enumerate}
    \item $\mathlarger{\int_{-\infty}^0 f(x) dx =\frac{1}{2}};$
    \item $\mathlarger{\int_{0}^\infty f(x) dx =\frac{1}{2}};$
    \item $\mathlarger{F(0)=\frac{1}{2}.} $
\end{enumerate}
\end{definition}
In the next definition, we define a partial ordering $\prec$ that describes when one probability distribution is of consistently greater spread than another.

\begin{definition}
\label{spread defn}
Let $f,g$ be two median-zero continuous pdf's. Say \emph{$g$ has greater volatility than $f$ everywhere}, denoted $f \prec g$ if $F(x)<G(x)$ for $x \in (-\infty,0)$ while $F(x)>G(x)$ for $x \in (0, \infty).$
\end{definition}

We make several remarks on this definition. First, this is a partial rather than a total ordering, as two arbitrary cdf's $F,G$ may frequently intersect for $x \in \mathbb{R}.$ Second, we note that this notation is only used in Section \ref{theoretical justification}. Third, this is quite a strong definition - while volatility may sometimes refer to only the second moment, this definition carries information about higher moments as well - broadly speaking if $f \prec g$ then $g$ has all higher moments greater than that of $f$. Finally, the utility of this definition is predicated upon the following proposition and corollary. For an example of two pdf's that fit the scope of both Proposition \ref{alt defn prop} and Corollary \ref{normal corollary}, we include Figure \ref{fig:normaldist}.

\begin{prop}
\label{alt defn prop}
Let $f,g$ be two median-zero continuous pdf's. Suppose there exist $\tau_1,\tau_2>0$ such that $f(x)>g(x)$ for $x \in (-\tau_1,\tau_2)$ but $f(x)<g(x)$ if $x<\tau_1$ or $x>\tau_2$. Then $f \prec g.$
\end{prop}
This proposition gives a means to establish that $f \prec g$ by quick inspection.

\begin{figure}
    \centering
\includegraphics[width=0.9\textwidth]{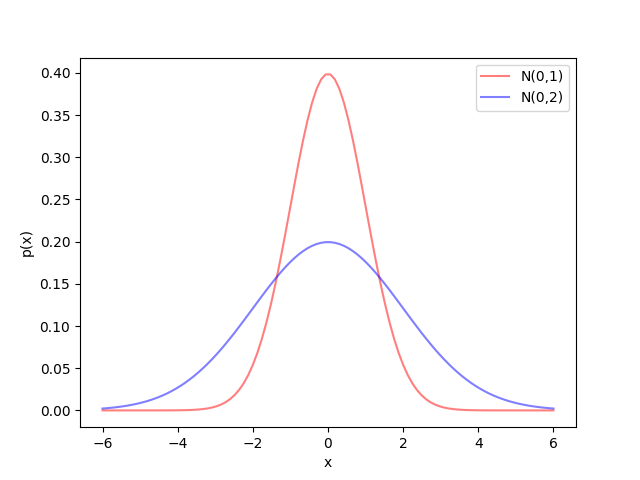}
    \caption{Normal distributions satisfying conditions of Proposition \ref{alt defn prop} and Corollary \ref{normal corollary}. These results prove that $N(0,1)\prec N(0,2)$.}
    \label{fig:normaldist}
\end{figure}

\begin{proof}
First, the continuity of $f$ and $g$ implies that $f(\tau_1)=g(\tau_1)$ and $f(\tau_2)=g(\tau_2)$. Next, let $x \leq \tau_1$. By integrating $f$ and $g$ over the interval $t \in (-\infty,x)$, it is immediate that $F(x)<G(x)$. Now, consider the interval $[\tau_1,0]$ and the function $h=F-G$ on this interval. We have established that $h(\tau_1)<0$ while $h(0)=0$ by the median-zero property. By the fundamental theorem of calculus, $h$ is differentiable on the interval with $h'(t)=f(t)-g(t)>0$ for $t \in (\tau_1,0).$ Thus $h$ is increasing on the interval $[\tau_1,0]$. Since $h(0)=0$, this implies $h(t)<0$ for $t \in (\tau_1,0)$, that is $F(t)<G(t)$ on the interval $t \in (\tau_1,0)$.

Together with the earlier fact that $F(x)<G(x)$ for $x \leq \tau_1$, we deduce that $F(x)<G(x)$ for all $x<0$. By an identical argument, we deduce $F(x)>G(x)$ for all $x>0$. That is, $f \prec g$, which concludes the proof.
\end{proof}

\begin{corollary}
\label{normal corollary}
Let $0<\sigma_1 < \sigma_2$ and let $f_1,f_2$ be the pdf's associated to normal distributions $N(0,\sigma_1), N(0,\sigma_2)$ respectively. Then $f_1 \prec f_2$.
\end{corollary}
\begin{proof}
We have $f_i(x) =\frac{1}{\sqrt{2\pi}\sigma_i} \exp(-x^2/2\sigma_i^2)$ for $i=1,2$. By simplifying, $f_1(x)>f_2(x)$ if and only if
\begin{align}
\label{eq:f1f2}
    x^2 < \frac{2\sigma_1^2 \sigma_2^2}{\sigma_2^2 - \sigma_1^2} \log \left( \frac{\sigma_2}{\sigma_1}\right).
\end{align}
Let $\tau$ be the square root of the right hand side of (\ref{eq:f1f2}). Then $f_1(x)>f_2(x)$ precisely for $x \in (-\tau,\tau)$, which establishes $f_1 \prec f_2$.

\end{proof}

The following lemma is a technical result for evaluating Wasserstein distances between pdf's that can be ordered with $\prec$.
\begin{prop}
\label{technicallemma}
Let $f<g$. Then the Wasserstein distance $W(f,g)$ can be computed as 
\begin{align}
    W(f,g)= \int_{-\infty}^0 (G-F)dx + \int_{0}^\infty (F-G) dx.
\end{align}
\end{prop}
\begin{proof}
In general, $W(f,g)$ may be computed \citep{DelBarrio} as 
\begin{align*}
    \int_{\mathbb{R}} |F - G| dx.
\end{align*}
In the case where $f \prec g$, $F-G$ is non-negative on $[0,\infty)$ and negative on $(\infty,0)$. Thus this integral can be expressed as a sum
\begin{align*}
 \int_{-\infty}^0 |F - G| dx + \int_{0}^\infty |F - G| dx = \int_{-\infty}^0 (G-F)dx + \int_{0}^\infty (F-G) dx, \end{align*}
 as required.
\end{proof}

\begin{theorem}
\label{main thm}
Suppose $f_1,\dots,f_n$ is a collection of median-zero pdf's that can be totally ordered by $\prec$. Then our method described in Section \ref{methodology} partitions the segments into intervals under $\prec$. That is, if $f_1 \prec \dots \prec f_n$ then the methodology necessarily segments into clusters $C_1=\{f_1,\dots,f_{n_1}\}, C_2=\{f_{n_1+1},\dots,f_{n_2}\},\dots,C_k=\{f_{n_{k-1}+1},\dots,f_{n_k}\}$. Moreover, our methodology is also resistant to deformations: if $f_1,\dots,f_n$ are a sufficiently small deformation away from totally ordered median-zero pdf's, then the result still holds.
\end{theorem}
\begin{proof}
First we assume the collection can be totally ordered under $\prec$. By relabelling, without loss of generality $f_1 \prec \dots f_n$. Then, if $i<j$,

\begin{align}
\label{eq3} W(f_i,f_j)&=    \int_{-\infty}^0 (F_j-F_i)dx + \int_{0}^\infty (F_i-F_j) dx\\
\label{eq4} &=\left(\int_{-\infty}^0 (F_j-F_1)dx + \int_{0}^\infty (F_1-F_j) dx \right) \\
\label{eq5} &-\left(\int_{-\infty}^0 (F_i-F_1)dx + \int_{0}^\infty (F_1-F_i) dx \right).
\end{align}
Thus, let $a_i=\mathlarger{\int_{-\infty}^0 (F_j-F_1)dx + \int_{0}^\infty (F_1-F_j)dx}$. Then $0=a_1<a_2<\dots < a_n$ and  $W(f_i,f_j)=a_j - a_i$ by (\ref{eq3}), (\ref{eq4}) and (\ref{eq5}) above. That is, the distance matrix $D$ between the segments $f_1,\dots,f_n$ is identical to a distance matrix between real numbers $a_1<\dots<a_n$. Spectral clustering applied to real numbers partitions them into intervals, that is $C_1=\{a_1,\dots,a_{n_1}\}, C_2=\{a_{n_1+1},\dots,a_{n_2}\},\dots,C_k=\{a_{n_{k-1}+1},\dots,a_{n_k}\}$ for some integers $1 \leq n_1<\dots<n_k=n$. Since $D_{ij}=a_j - a_i$ the results of spectral clustering are identical to applying it on the corresponding real numbers. Hence the segments are partitioned into corresponding clusters $C_1=\{f_1,\dots,f_{n_1}\}, C_2=\{f_{n_1+1},\dots,f_{n_2}\},\dots,C_k=\{f_{n_{k-1}+1},\dots,f_{n_k}\}$ as required.

Finally, if the initial collection of pdf's were not precisely median-zero and totally ordered under $\prec$ but a sufficiently small deformation away from a median-zero totally ordered collection, then the result still holds. By the continuity of the Wasserstein distance, the distance matrix $D$ would be a small deformation away from a distance matrix between real numbers. By the continuity of the procedures within spectral clustering, once again the $f_i$ would be partitioned into clusters of increasing volatility.
\end{proof}

\section{Results}
\label{results}
\subsection{Synthetic data experiments}

In this section, we validate our method on 200 synthetic time series. In two distinct experiments, we randomly generate 100 times series each with artificially pronounced breaks in volatility by concatenating different segments. Each segment is randomly drawn from various data generating processes and randomly chosen between 200 and 300 in length, together with added Gaussian random noise $\delta,\epsilon$ to ensure none of the data generating processes are identical. Thus, we create 200 unique time series where a ground truth number of clusters and the segment memberships are known - to these we apply our methodology to validate it. In a small number of instances, our methodology does not correctly identify the number of segments at the change point detection phase - we term this a \emph{mismatch}. Excluding these, when the change point detection step correctly identifies the number of segments, we evaluate the quality of the clustering using the Fawlkes-Mallows index \citep{Fowlkes1983,Halkidi2001}:
$$FMI = \sqrt{ \frac {TP}{TP+FP} \cdot \frac{TP}{TP+FN}},$$
where $TP, FP, FN$ are the number of true positives, false positives and false negatives, respectively. The score is bounded between 0 and 1, where 1 indicates a perfect cluster assignment. This FMI can only be applied when the change point detection step correctly identifies the number of segments, which we call \emph{correctly matched}.

In the first experiment, we draw from five different normal distributions, as follows:
\begin{gather*}
X_1 \sim \mathcal{N}(0+\delta, (0.25+\epsilon)^2) \\
X_2 \sim \mathcal{N}(0+\delta, (0.5+\epsilon)^2) \\
X_3 \sim \mathcal{N}(0+\delta, (1+\epsilon)^2) \\
X_4 \sim \mathcal{N}(0+\delta, (2+\epsilon)^2) \\
X_5 \sim \mathcal{N}(0+\delta, (4+\epsilon)^2)
\end{gather*}
They must be approximately mean-centred to mimic the properties of the log returns time series and for the Mood test to detect change points in the variance, as detailed in Appendix \ref{cpa appendix}.

In this first experiment, there were 4 out of 100 mismatches. In all these instances, the algorithm was off by just one in detecting the correct number of segments. Among the 96 remaining time series, we are able to apply the two different methods to select the number of clusters and apply spectral clustering, as described in Section \ref{spectral clustering}. We refer to the two methods as the eigengap method to generate $k^e$ and the gradient descent method to generate $k^{ZP}$, as defined in \ref{spectral clustering}. To validate our methodology, we record both the mean FMI score for these two methods as well as smoothed histograms over the 96 experiments.

\begin{figure}[htbp]
\centering
\subfigure[Synthetic time series change points and detection times\label{fig:syn_time_cp}]{%
\resizebox*{9.7cm}{!}{\includegraphics{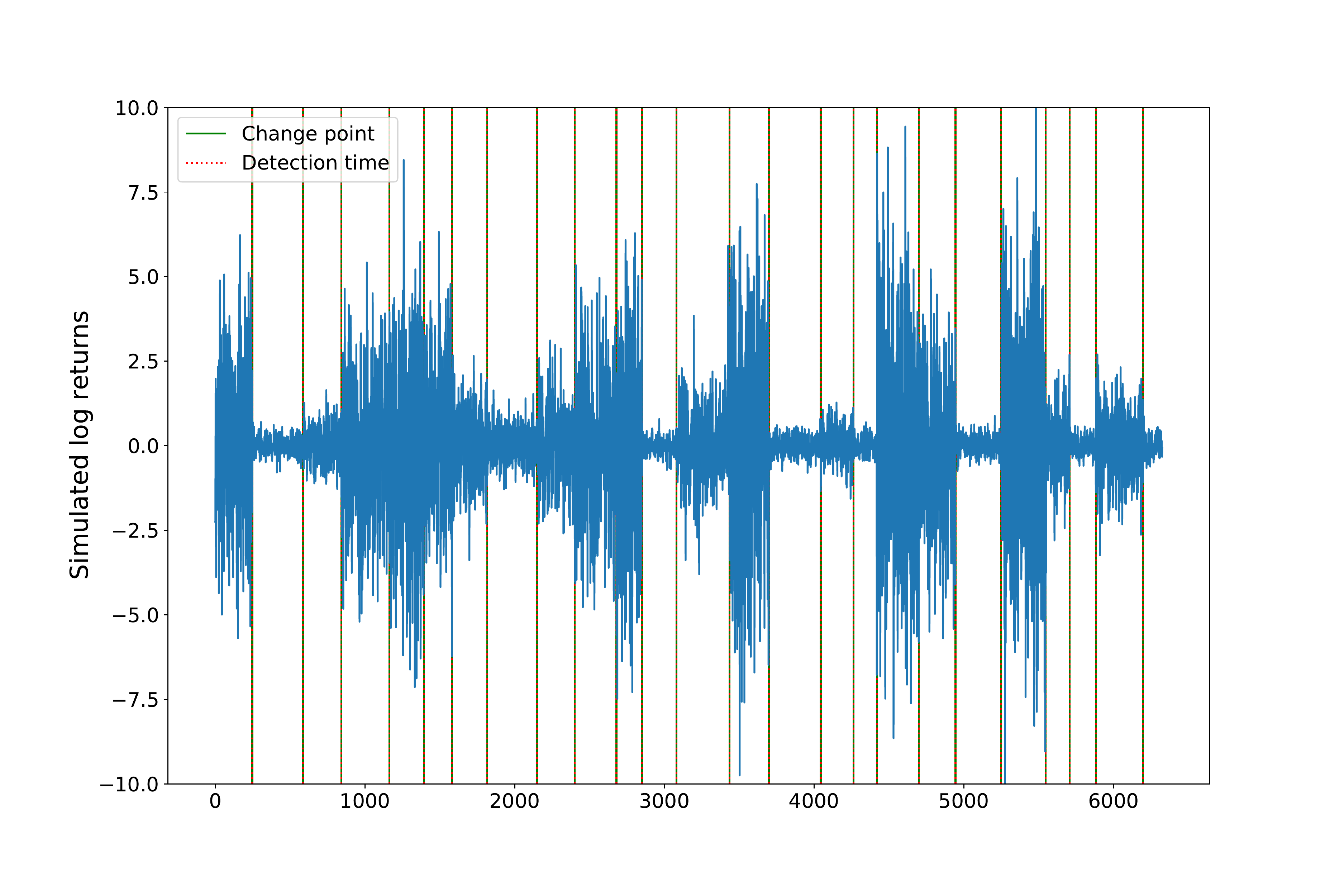}}}
\hspace{5pt}
\subfigure[Kernel density estimation plots of synthetic distributions, forming five clusters\label{fig:syn_cl_d}]{%
\resizebox*{9.7cm}{!}{\includegraphics{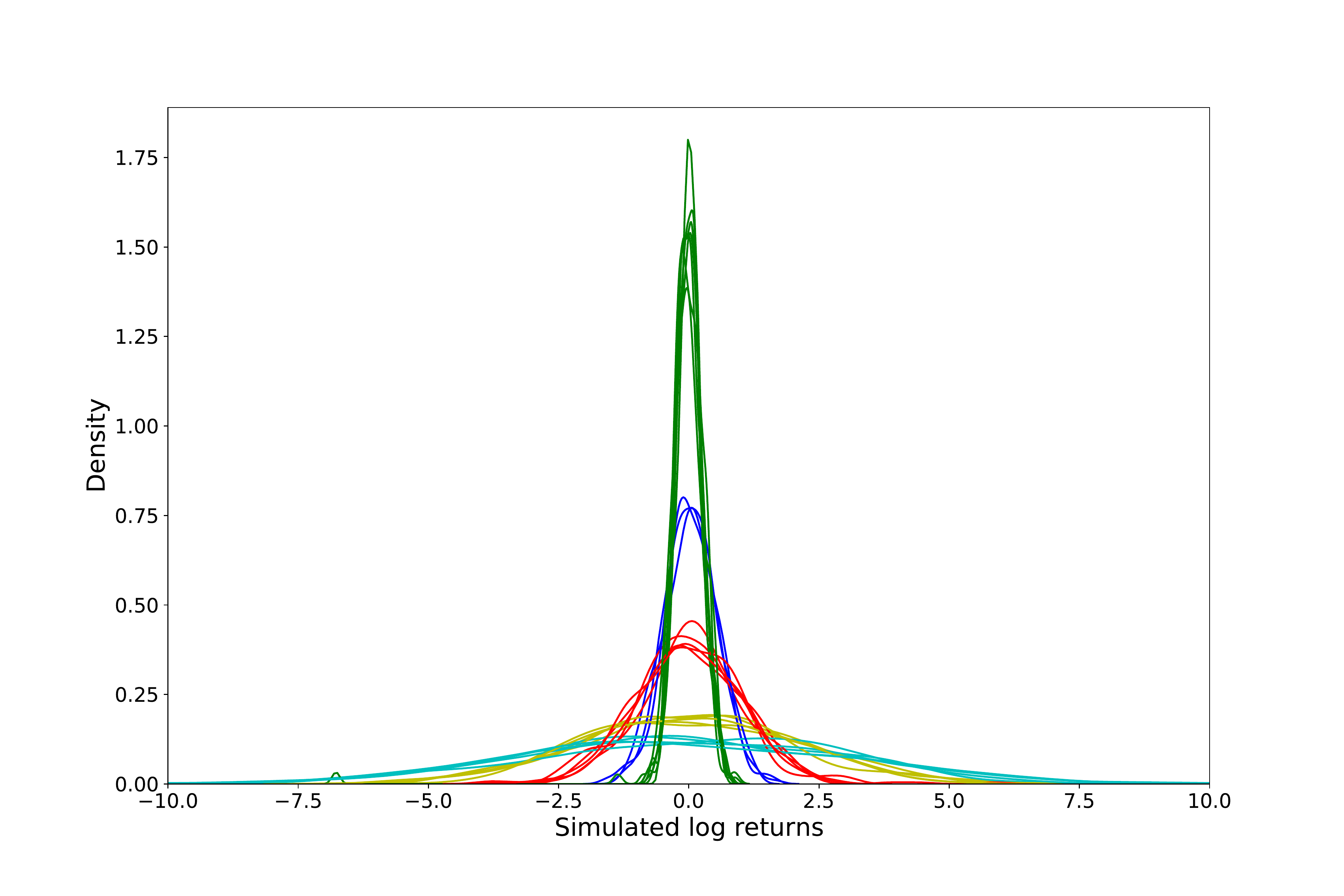}}}
\subfigure[Five determined volatility regimes\label{fig:syn_cl_t}]{%
\resizebox*{9.7cm}{!}{\includegraphics{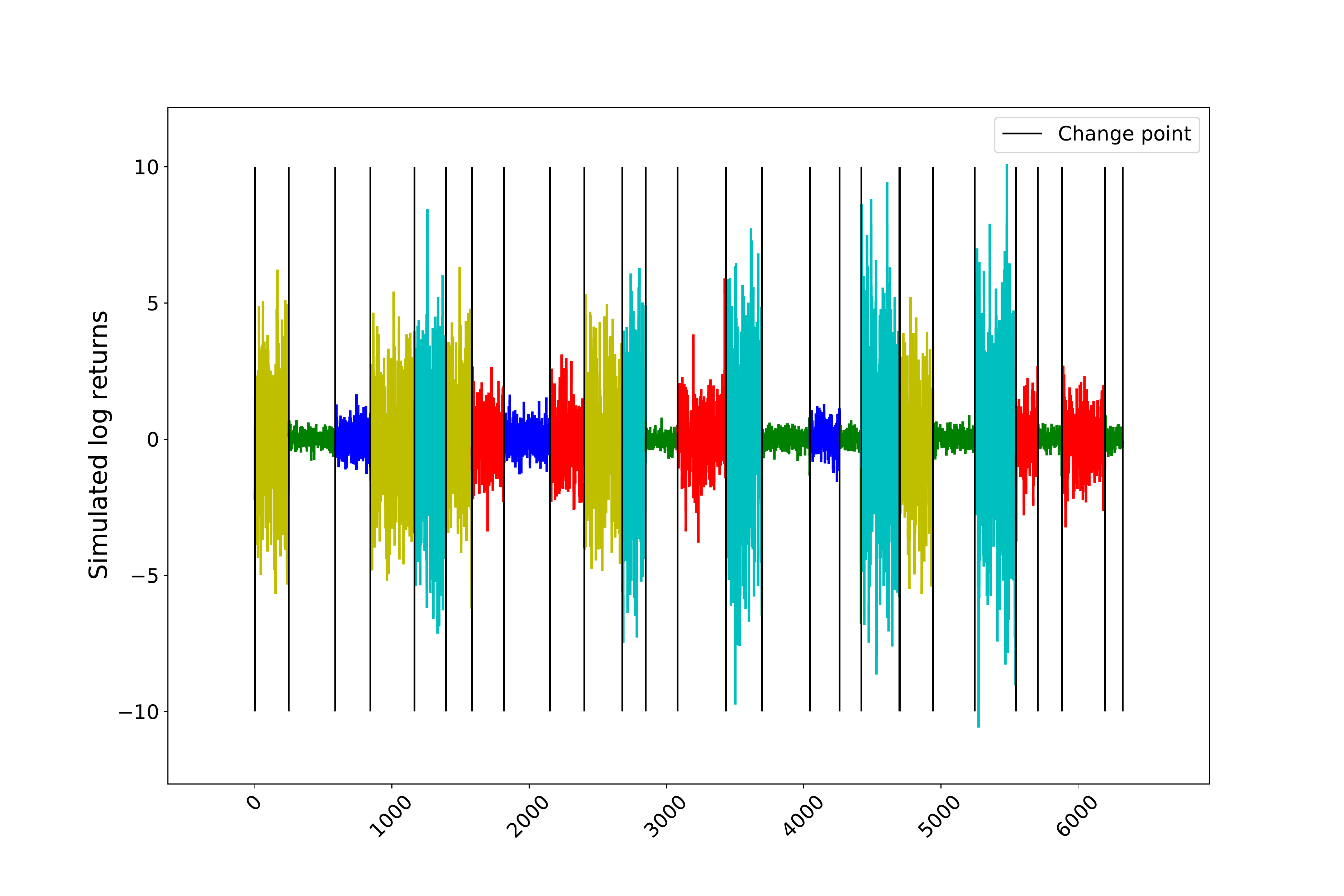}}}
\hspace{5pt}
\caption{Synthetic data experiment for normal distributions (five regimes)} 
\label{fig:syn}
\end{figure}

In Figure \ref{fig:syn}, we present one example of such a time series. Figure \ref{fig:syn_time_cp} displays the change point partitions; the detection times, as described in Appendix \ref{cpa appendix}, are not visible. Figure \ref{fig:syn_cl_d} displays the kernel density estimations of the distributions, coloured according to their membership in five detected clusters. Figure \ref{fig:syn_cl_t} shows the final clustering of the segments of the synthetic time series. This whole procedure correctly identifies the change in variance, as well as the existence of five regimes (clusters) of volatility.

Across all 96 correctly matched iterations of this first experiment, the mean FMI score for the eigengap method ($k^e$) was 0.93, while the mean FMI score for the gradient descent method ($k^{ZP}$) was 0.86. The smoothed histogram in Figure \ref{fig:normal_fmi} shows that the eigengap method has systematically higher FMI scores.

For the second experiment, we repeated the same procedure, however we drew from five Laplace distributions:
\begin{gather*}
X_1 \sim \mathcal{F_L}(0+\delta, 0.25+\epsilon) \\ 
X_2 \sim \mathcal{F_L}(0+\delta, 0.5+\epsilon) \\
X_3 \sim \mathcal{F_L}(0+\delta, 1+\epsilon) \\
X_4 \sim \mathcal{F_L}(0+\delta, 2+\epsilon) \\
X_5 \sim \mathcal{F_L}(0+\delta,4+\epsilon)
\end{gather*}
In this second experiment, there were 11 out of 100 mismatches. Once again, the algorithm was off by just one in detecting the correct number of segments in all instances. Across all 89 correctly matched iterations of this second experiment, the mean FMI score for the eigengap method ($k^e$) was 0.94, while the mean FMI score for the gradient descent method ($k^{ZP}$) was 0.85. Again, the smoothed histogram in Figure \ref{fig:lap_fmi} shows that the eigengap method has systematically higher FMI scores than the alternative.

\begin{figure}[htbp]
\centering
\subfigure[Distribution of FMI scores over 96 correctly matched instances \label{fig:normal_fmi}]{%
\resizebox*{14cm}{!}{\includegraphics{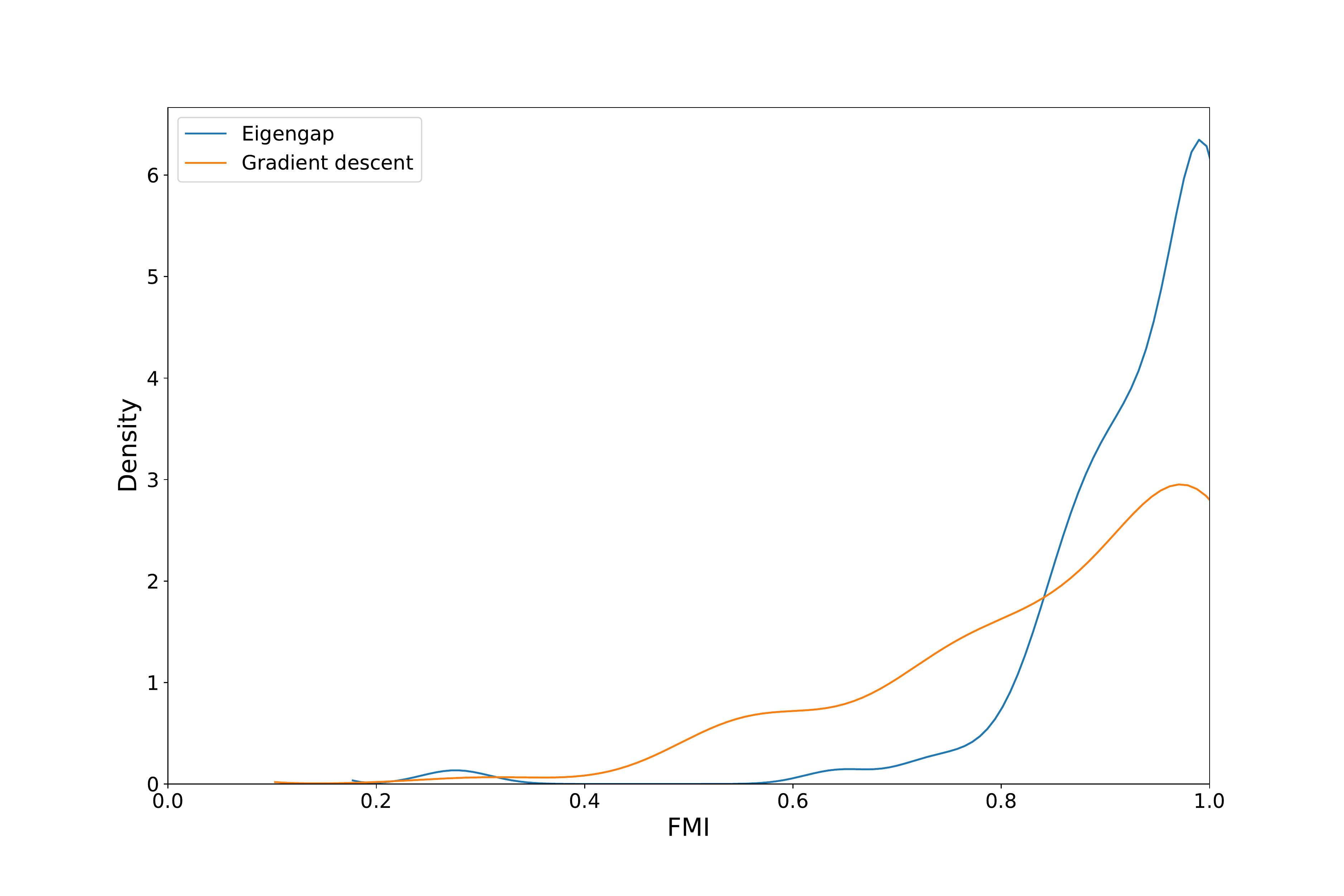}}}\hspace{5pt}

\subfigure[Distribution of FMI scores over 89 correctly matched instances\label{fig:lap_fmi}
]{%
\resizebox*{14cm}{!}{\includegraphics{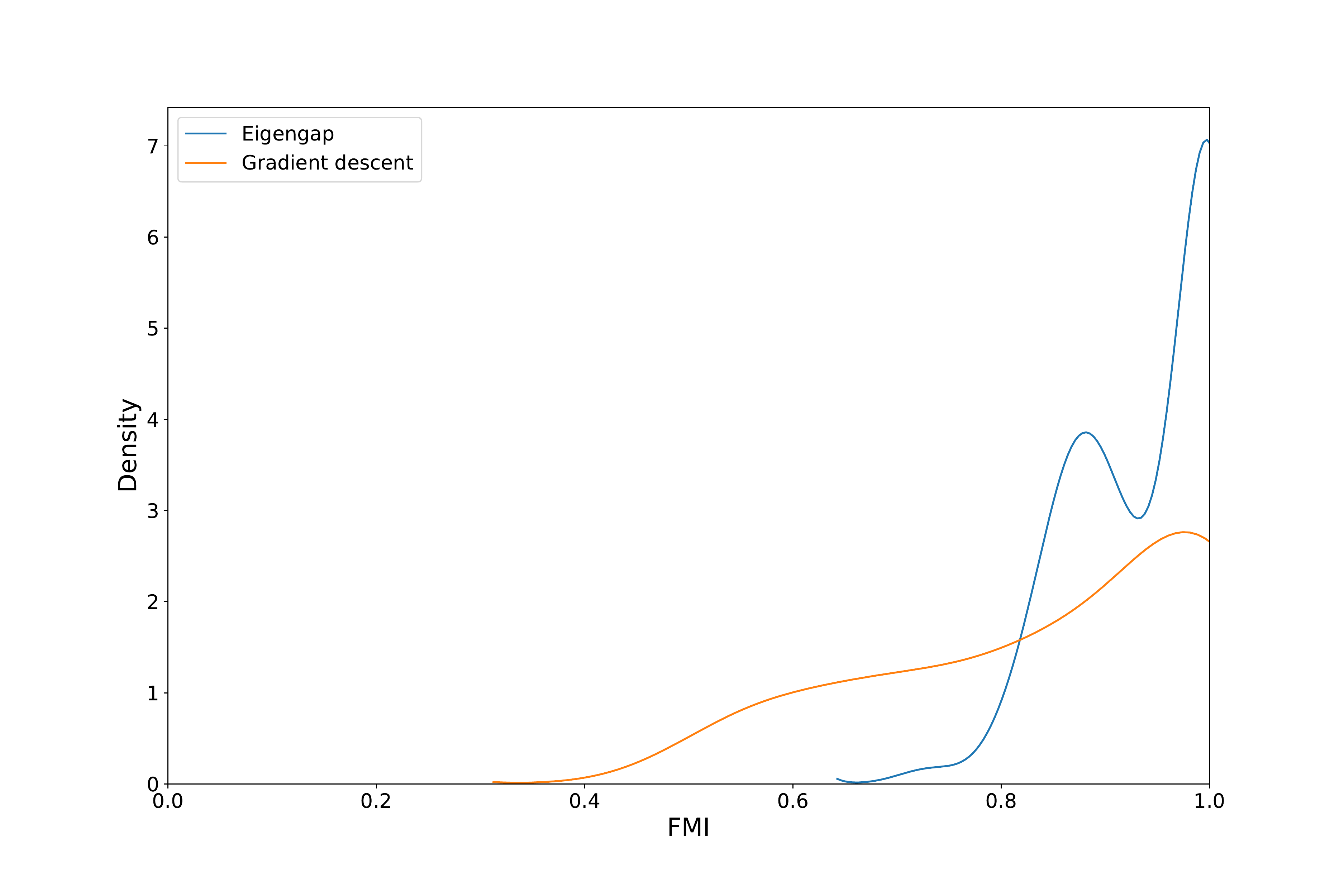}}}
\caption{Smoothed histograms for FMI scores from the eigengap and gradient descent method for (a) synthetic Gaussian and (b) synthetic Laplace experiments}
\label{fig:fmi}
\end{figure}

Manually examining the mismatches shows that they tend to occur when the one segment is sandwiched between two segments drawn from the next most similar distribution, for example when a $\mathcal{N}(0,1)$ segment is sandwiched between two $\mathcal{N}(0,2)$ segments. When this happens, the algorithm tends to split the segments into four parts instead of three, leading to the overestimation. It is also worth noting that we repeated both sets of experiments using the Bartlett test instead of the Mood test to help reduce mismatches. While both tests performed similarly in the first experiment, the Bartlett test had a much higher mismatch rate in the second experiment, leading to much higher estimations of the number of clusters. This is to be expected as the Barlett test is a parametric test for changes in normal distributions. The better performance of the Mood test in both the normal and Laplace experiments indicates that it is better suited to real-world log returns data, which are not necessarily normally distributed. Having observed few mismatches and better FMI scores with the eigengap method, we proceed by exclusively applying the eigengap method to real data in subsequent sections.

\subsection{SPY}
\label{SP500}

In this section, we apply our method to SPY, an ETF of the S\&P 500, and analyze our volatility clustering results. We study adjusted closing price data from Yahoo! Finance, \url{https://finance.yahoo.com}, from 1 January 2008 to 31 December 2020, and compute the log returns before applying our methodology.

Our algorithm finds five volatility clusters over the 13-year period studied, seen in Figure \ref{fig:spy}. The cyan cluster is associated with extreme market behaviours such as the worst part of the global financial crisis (GFC). The next most volatile cluster is displayed in yellow, and the majority of the GFC, the August 2011 crash and the entirety of 2020 are grouped in this period. The red cluster contains the start of the GFC, the 2010 flash crash and the 2015-16 sell off and the US/China Trade War. Finally the blue and green clusters display more stable economic periods. Our trading strategy is predicated upon the idea that we associate the most recent window with a period in the past that is most similar. Details are described in Section \ref{trading strat}. We observe that a change point is detected between the 17th and 18th segment of Figure \ref{fig:spy_cl_t}, and yet there was no regime change in volatility at this time. This can occur when the distributions are different, but not different enough to warrant an entire regime change.

\begin{figure}[htbp]
\centering
\subfigure[SPY change points and detection times \label{fig:spy_time_cp}]{%
\resizebox*{9.7cm}{!}{\includegraphics{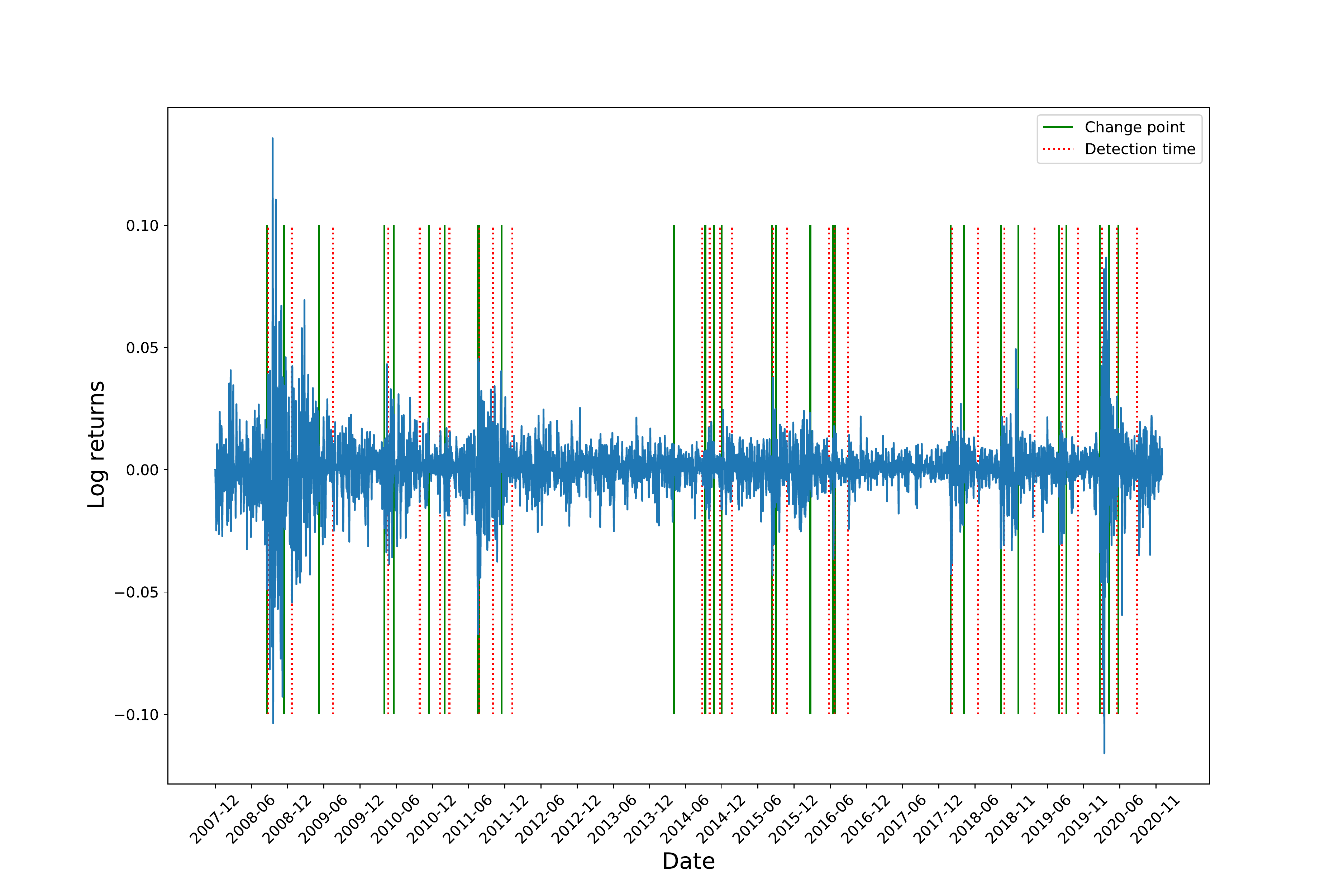}}}\hspace{5pt}
\subfigure[Kernel density estimation plots of SPY distributions, forming five clusters \label{fig:spy_cl_d}]{%
\resizebox*{9.7cm}{!}{\includegraphics{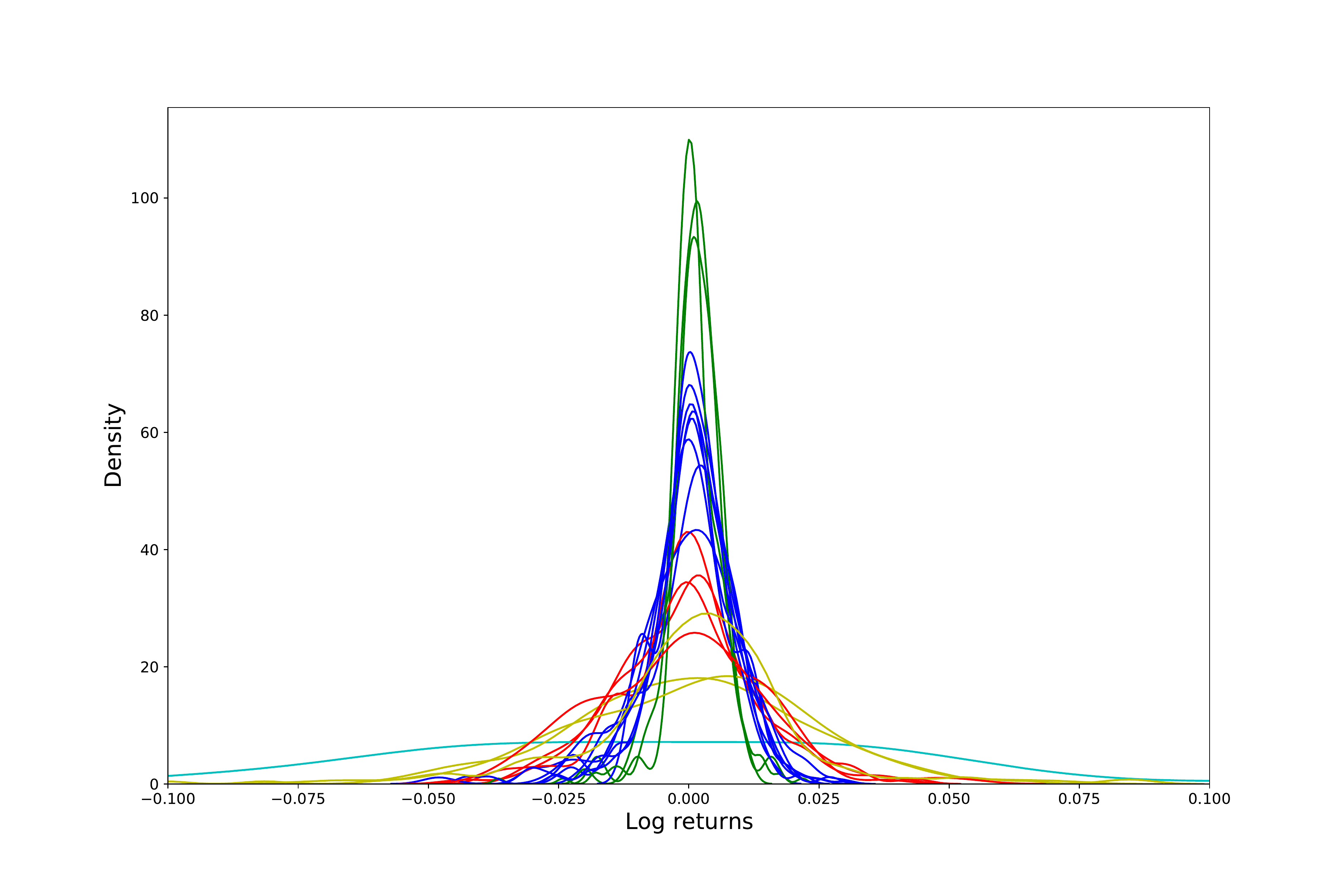}}}
\subfigure[Five determined volatility regimes for SPY \label{fig:spy_cl_t}]{%
\resizebox*{9.7cm}{!}{\includegraphics{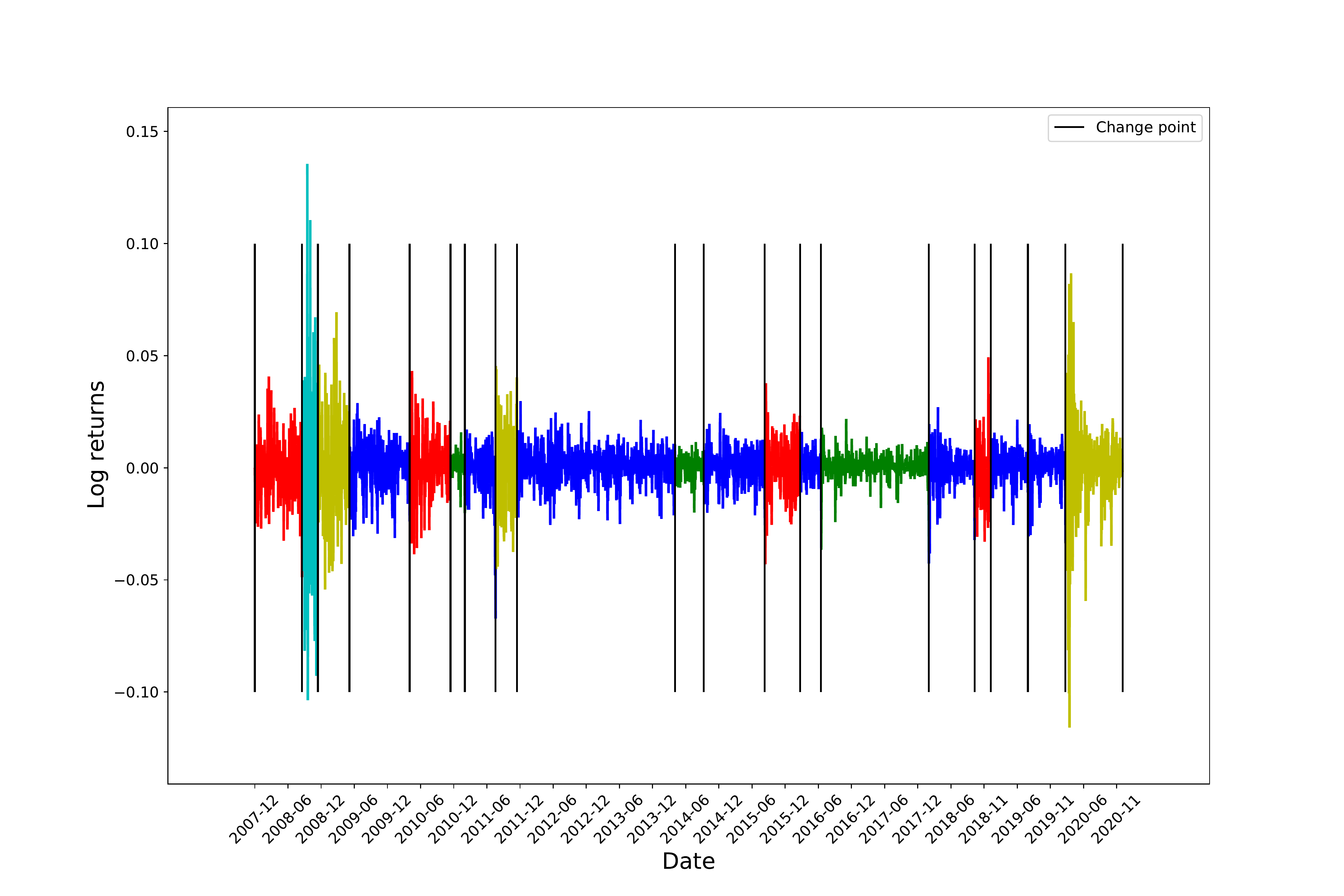}}}\hspace{5pt}
\caption{Volatility clustering results for SPY} \label{fig:spy}
\end{figure}

\subsection{Empirical results on various asset classes}

In this section, we apply our methodology to a collection of asset classes including equity indices, individual stocks, ETFs and currency pairs. For each asset, we report the learned number of segments and clusters. This information is used to design the trading strategy. We display these results in Tables \ref{table:index_results}, \ref{table:big5_results}, \ref{table:etfs} and \ref{table:currs}, respectively. 

Table \ref{table:index_results} displays results for the S\&P 500, Dow Jones, Nikkei, FTSE, ASX 200 indices. Table \ref{table:big5_results} displays results for large firms including: MSFT (Microsoft), APPL (Apple), AMZN (Amazon), GOOG (Alphabet) and BRK-A (Berkshire Hathaway Class A). Table \ref{table:etfs} displays results for the collection of ETFs studied, including: RYT (Invesco S\&P 500 Equal Weight Technology ETF), GLD (SPDR Gold Shares), XLF (Financial Select Sector SPDR Fund), IJS (iShares SP Small-Cap 600), MSCI (World Index Fund). Finally, Table \ref{table:currs} displays results for several currency pairs. In Appendix \ref{app4}, we provide figures displaying various time series' partitioned volatility regimes, and the respective clusters in which they are associated.

\begin{table}[htpb]
\tbl{Volatility cluster structure for major indices}
{ \begin{tabular}{c c c c} 
 \toprule
 Ticker & No. Segments & No. Clusters & Cluster Sizes \\ [0.5ex] 
 \midrule
 S\&P 500 &19& 5 & 3, 8, 4, 3, 1 \\ 
 Dow Jones &19 &  5 & 5, 6, 3, 3, 2 \\ 
 Nikkei 225   & 17 & 2 & 6, 11  \\ 
 MSCI & 15 & 4 & 4, 5, 4, 2 \\ 
 \bottomrule
\end{tabular}}
\label{table:index_results}
\end{table}

\begin{table}[htpb]
\tbl{Volatility cluster structure for large firms}
{\begin{tabular}{l c c c c} 
\toprule
Ticker  & No. Segments & No. Clusters & Cluster Sizes \\
[0.5ex] 
\midrule
MSFT   & 15 & 4 & 6, 7, 1, 1 \\ 
AAPL &  15 & 3 & 4, 5, 6 \\ 
AMZN  & 14 & 2  & 7, 7  \\
GOOG  & 13 & 2 & 6, 7  \\ 
BRK-A & 15 & 4 & 5, 4, 3, 3 \\ 
\bottomrule
\end{tabular}}
\label{table:big5_results}
\end{table}

\begin{table}[htpb]
\tbl{Volatility cluster structure for popular ETFs}
 {\begin{tabular}{l c c c c} 
 \toprule
 Ticker  & No. Segments &No.  Clusters & Cluster Size  \\ [0.5ex] 
 \midrule
 GLD    & 15 & 3 & 4, 6, 5 \\ 
 XLF & 21 & 3 & 9, 8, 4 \\ 
 IJS  & 18 & 6 & 3, 4, 4, 3, 2, 2 \\ 
 RYT & 17 & 4 & 4, 4, 5, 4 \\ 
\bottomrule
\end{tabular}}
\label{table:etfs}
\end{table}

\begin{table}[htpb]
\tbl{Volatility cluster structure for currency pairs}
 {\begin{tabular}{l c c c c} 
\toprule Ticker  & No. Segments &No.  Clusters & Cluster Size  \\ [0.5ex] 
 \midrule
 AUD/USD  & 14 &  3 & 5, 5, 4 \\ 
 EUR/USD & 14 & 2 & 8, 6 \\ %
 GBP/USD  & 13 & 4 & 3, 5, 3, 2 \\ %
 NZD/USD & 12 & 3 & 4, 5, 3 \\ %
\bottomrule
\end{tabular}}
\label{table:currs}
\end{table}

\subsection{Sensitivity of methodology}
\label{sec:sensitivity}

All results in Section \ref{results} are obtained over a 13-year period from 2008 to 2020. In this brief section, we discuss the effects of extending or subdividing our window of analysis, which will be particularly relevant for Section \ref{trading strat}. When the period is extended or subdivided, results may change in numerous ways. First, the change point algorithm used (described in Appendix \ref{cpa appendix}) feeds in one data point in at a time. It may continually detect new change points, subject to a lower bound of 30 days before a new change point is permitted. Thus, extending the period of analysis will necessarily only increase the number of determined segments.

On the other hand, suppose the total period is subdivided into smaller intervals. For example, suppose a longer period $a \leq t \leq b$ is partitioned into $a \leq t \leq c$ and $c \leq t \leq b$. Broadly speaking, the number of change points in $[a,b]$ will usually be the sum of the change points in $[a,c]$ and $[c,b]$. This does not necessarily hold all the time, for example if a change point in $[a,b]$ is detected close to $t=c$; then it would feature in $[a,b]$ but potentially in neither $[a,c]$ nor $[c,b]$. Translating this into the number of segments, this means that $n_{[a,b]}$ is usually $n_{[a,c]} +n_{[c,b]} - 1$, where $n_I$ is the number of segments observed within an interval $I$. However, this may not hold exactly.

The total number of clusters, however, is more difficult to predict regarding how it will behave with different period lengths. For example, as the analysis window is extended, more segments may be observed, yielding a larger space of segments to be clustered. However, as a space (in this case, consisting of probability densities) grows in cardinality, its determined number of clusters may change unpredictably. Usually, a larger cardinality of points in a high-dimensional space should produce a greater number of determined clusters. Yet this is not always the case. Consider a hypothetical example of a space neatly divided into three clusters (indexed 1, 2, and 3). If enough data points are added to ``bridge'' clusters 2 and 3, then clustering may subsequently detect only two clusters (1 and $2\cup 3$). In our implementation, we broadly notice that smaller periods of analysis produce fewer clusters. This will be relevant in the precise trading methodology of Section \ref{trading strat}.

Finally, we remark that even two adjacent segments may lie in the same determined cluster. Indeed, this was demonstrated in Section \ref{SP500}. Fortunately, this works well in the context of our methodology. Synthetic experiments demonstrate that the change point detection algorithm identifies more change points than the ground truth time series in some circumstances. Our clustering overcomes this potential issue of oversensitivity. In the event that an erroneous change point partitions two adjacent segments with similar statistical properties, the clustering method would determine that these segments exist in the same cluster, ameliorating the false positive of the change point algorithm.

\subsection{Discussion of findings}

Our empirical results demonstrate substantial heterogeneity in the volatility structure among the asset classes studied. The number of clusters ranges between 2 and 6, with the most common number of volatility regimes determined to be 3 or 4. The number of segments belonging to each clustering regime is relatively consistent among all asset classes. The one exception is at the height of the GFC, where a single segment is associated with one cluster. When a shorter time period is used, two regimes are consistently identified, which is consistent with prior findings \citep{Guidolin2011}.

Despite heterogeneity in the number of segments and clusters, we can still identify similarities between asset classes. For instance, the S\&P 500 and the Dow Jones both exhibit periods of high volatility in March 2010, April 2011, and late 2018. All but several asset classes exhibit a volatile clustering regime at the height of the GFC. Typically, this period belongs to the same cluster as the COVID-19 market crash of 2020 - highlighting similar market structure during periods of severe crisis \citep{James2021}. In addition, all five of the listed firms experienced a similarly volatile window associated with the US/China trade war in late 2018. By contrast, the ETFs under analysis did not exhibit similarly volatile periods - most likely due to their varied asset class mix.

In their typical formulation, regime-switching models require assumptions regarding the number of regimes and candidate distributions \textit{a priori}. They are often criticized for this highly parameterized structure \citep{Guidolin2011, Ang2012}. One meaningful advantage of our method, is the ability to account for extreme economic events and market crises with the flexibility of our method. No assumptions regarding the number of regimes or data generating process are required - indeed, we have validated our method both on Gaussian and non-Gaussian distributions. For those wishing to implement parametric regime-switching models, the methodology proposed in this work could be used as an accompaniment to algorithmically determine an appropriate number of regimes for various asset classes. As a further application, in the proceeding section we demonstrate how these results can be used to inform asset allocation decisions in the context of a dynamic trading strategy.


\section{Application of results: trading strategy}
\label{trading strat}

In recent times, passive investing has gathered more asset inflows than active investment management. In particular, index funds and ETFs that track major indices such as the S\&P 500 are a popular way of attaining broad market exposure for investors. We apply our analysis to the S\&P 500 index to determine a dynamic trading strategy that can simultaneously benefit from the index's appreciation while minimising risk. In Section \ref{SP500}, we determined that the S\&P 500 has two distinct volatility regimes, captured in two distinct clusters of volatility periods. Our contrived trading strategy is to buy and hold SPY, a tracker of the S\&P 500, in low volatility periods, and then flee to the safe haven of GLD, a gold bullion tracker, in high volatility periods. We improve on the previous work of \citet{Nystrup2016}, who uses a live implementation of the rank test to move away from the S\&P 500. This method has two drawbacks: first, as noted in Section \ref{results}, a change point does not necessarily indicate a change in regime; secondly, their method has an unpredictable delay in registering the change point, as discussed in Appendix \ref{cpa appendix}. 

Instead, we implement a dynamic procedure with a 4-year sliding window. Model parameters are learned within the prior window, and then applied to the proceeding four years of data. Suppose our algorithm begins with years 0 : 4. First, we analyze the SPY over the prior 4-year period, years -4 : 0. We determine the cluster structure of the distribution segments of the SPY over this prior period. To make investment decisions in the current period of 0 : 4 years, we try to match the present distribution with the most similar distribution in the prior window. We combine metric geometry and unsupervised learning for this purpose, minimizing a metric to one of an existing set of candidate segments. Specifically, we examine the present local distribution of the last $n$ days, where $n$ is a learned parameter, and determine the minimal distance between the present local distribution and the set of segment distributions of the prior 4-year period. We call $n$ the \emph{look back length} of the procedure. If this closest prior segment lies in (one of) the most volatile class of past distributions (characterized by greatest variance), we determine that the local distribution is volatile, and allocate all capital toward gold. This method works even if more than $2$ volatility clusters are found during the previous window. For example, if $r$ volatility regimes are determined to exist, the algorithm could avoid SPY if the current period is matched to the most volatile $\lceil r/2 \rceil$ such regimes, ordered by variance. Fortunately, in our implementation, the number of volatility regimes in the 4-year windows is always 2.

The parameter $n$ is optimized relative to the -4 : 0 year window. Specifically, having determined the cluster structure, $n$ is chosen to optimize the \emph{Sharpe ratio}, a well-established measure of risk-adjusted returns, when testing over that window. We optimize $n$ over a range $20 \leq n \leq 30$, that is, 4 to 6 trading weeks. Thus, $n$ is learned in this prior window and then used in the algorithm in the subsequent window. The window is then successively slid forward four years, and the process repeats. That is, model parameters estimated on years 0 : 4 are used to forecast in years 4 : 8, and so on.

There are several reasons for our choice of 4-year windows. First, it is a suitable compromise in the length of the training data. If the training period were too short, we could erroneously capture transient behaviours as persistent market dynamics. By contrast, if the training period were too lengthy, our strategy may fail to prioritize more recent and relevant dynamics. Specifically regarding our use of clustering, our trading strategy always makes decisions (in the present) having been trained on data that was sufficiently recent to learn up-to-date behaviour of volatility regimes, but sufficiently long to observe non-trivial clustering results. In addition, this adaptive sliding window technique allows us to convincingly validate the long-run performance of our trading strategy. Second, as discussed in Section \ref{sec:sensitivity}, the number of volatility regimes generally increases with the window length. Across 4-year training periods, our algorithm always determines $r=2$ to be the optimal number of clusters when we examine SPY in its partitioned chunks,  as shown in Table \ref{table:SPY4}. This makes the decision process of avoiding volatile regimes simpler and less ambiguous - if the current distribution is matched (by minimal distance) to the more volatile (of two) class of distributions, we allocate all capital towards gold. Third, a 4-year period is suitable as equity markets follow 4-year cycles, associated with the cyclicality of Kitchin cycles \citep{Korotayev} and the US presidential election \citep{Gartner1995}.

We analyze the strategy's performance in a period from immediately prior to the global financial crisis (GFC), up to the present day. Accordingly, our initial backtest period of -4 : 0 is 2004-2008, while our first period of trading, years 0 : 4, is 2008-2012. We compare the performance of our dynamic trading strategy with three other strategies: holding SPY, holding GLD, and a baseline strategy holding an equal split between the two. We use six common validation metrics to evaluate and compare our trading strategy. 

\begin{enumerate}
    \item Annualized return (AR): the total return a strategy yields relative to the time the strategy has been in place.
    \item The overall standard deviation (SD) of the portfolio.  
    \item Sharpe ratio (SR): 
    a common measure of risk-adjusted return. Unfortunately, this penalizes both upside and downside volatility. Some strategies with strong annualized returns may have lower Sharpe ratios due to erratic, yet positive return profiles.
    
    \item Maximum drawdown (MD): an alternative penalty function capturing the maximum peak to trough trading loss. 
    
    \item Sortino ratio (SoR): 
    an alternative measure of risk-adjusted return that only penalizes downside deviation in the denominator.
    \item Calmar ratio (CR): a measure of risk-adjusted returns that penalizes the maximum realized drawdown over some candidate investment period. 
\end{enumerate}

If our trading strategy were applied among a collection of specific assets rather than an index (for example, a group of components of the index as featured in Table \ref{table:big5_results}), the trading strategy could attain greater expected returns and higher risk-adjusted return ratios. In particular, while an index has a mean-reverting effect with respect to a large collection of stock returns, a smaller collection of carefully selected stocks may provide higher expected returns. Should the group of stocks be greater than $\sim30$, unsystematic (stock-specific) risk is diversified away for both the index and the portfolio of stocks. Both portfolios would have a risk component mostly comprised of systematic risk. This could lead to higher risk-adjusted returns for the trading strategy.

\begin{table}[htpb]
\tbl{Volatility cluster structure of SPY in four year increments}
{\begin{tabular}{l c c c c} 
 \toprule
 Ticker  & No. Segments & No. Clusters & Cluster Sizes \\ [0.5ex] 
 \midrule
 2000-2004   & 7 & 2 & 4,3 \\ 
 2004-2008 &  4 &  2 & 3,1 \\
 2008-2012  & 9 & 2  & 4,5  \\
 2012-2016  & 4 &2 & 3,1 \\
 2016-2020 & 8 & 2 & 5,3\\ 
 2020-2021   & 3 & 2  & 2,1  \\
\bottomrule
\end{tabular}}
\label{table:SPY4}
\end{table}

\subsection{Model performance: 2008-2020}

Implementing our trading strategy between January 2008 and December 2020 would have been highly successful for both risk-averse and risk-on investors. Seen in Table \ref{table:2020bt} and Figure \ref{fig:2020bt}, the strategy consistently outperformed the S\&P 500 index and overall generated annualized returns of 9.4\%. The S\&P 500 returned 7.5\% while the static baseline strategy returned 7.6\%. The strategy clearly generates alpha by its dynamic nature, automatically detecting market regimes and allocating capital successfully. This entire period can broadly be characterized as a bull market, and yet features several severe market shocks; the strategy's consistent performance demonstrates its robustness to varied market dynamics. Figure \ref{fig:2020pos} shows the positions held by the strategy.

\begin{table}[ht]
\tbl{Validation metrics: January 2008 - December 2020}
{\begin{tabular}{l c c c c c c c} 
\toprule
Strategy & AR & SD & SR & MD & SoR & CR  \\ [0.5ex] 
\midrule
Hold SPY & 0.075 & 0.013 &  0.45  & 0.53 & 0.63 & 0.14 \\ 
Hold GLD & 0.057 & 0.011 &  0.40  & 0.46 & 0.56 & 0.13 \\ 
 Baseline & 0.076  & 0.0087 &  0.33  & 0.60 & 0.85 & 0.23 \\ 
\textbf{Dynamic} & \textbf{0.094}  & \textbf{0.010} &  \textbf{0.63}  & \textbf{0.27} & \textbf{0.89} & \textbf{0.35} \\
\bottomrule
\end{tabular}}
\label{table:2020bt}
\end{table}

\begin{figure}[htpb]
\centering
\includegraphics[width=\textwidth]{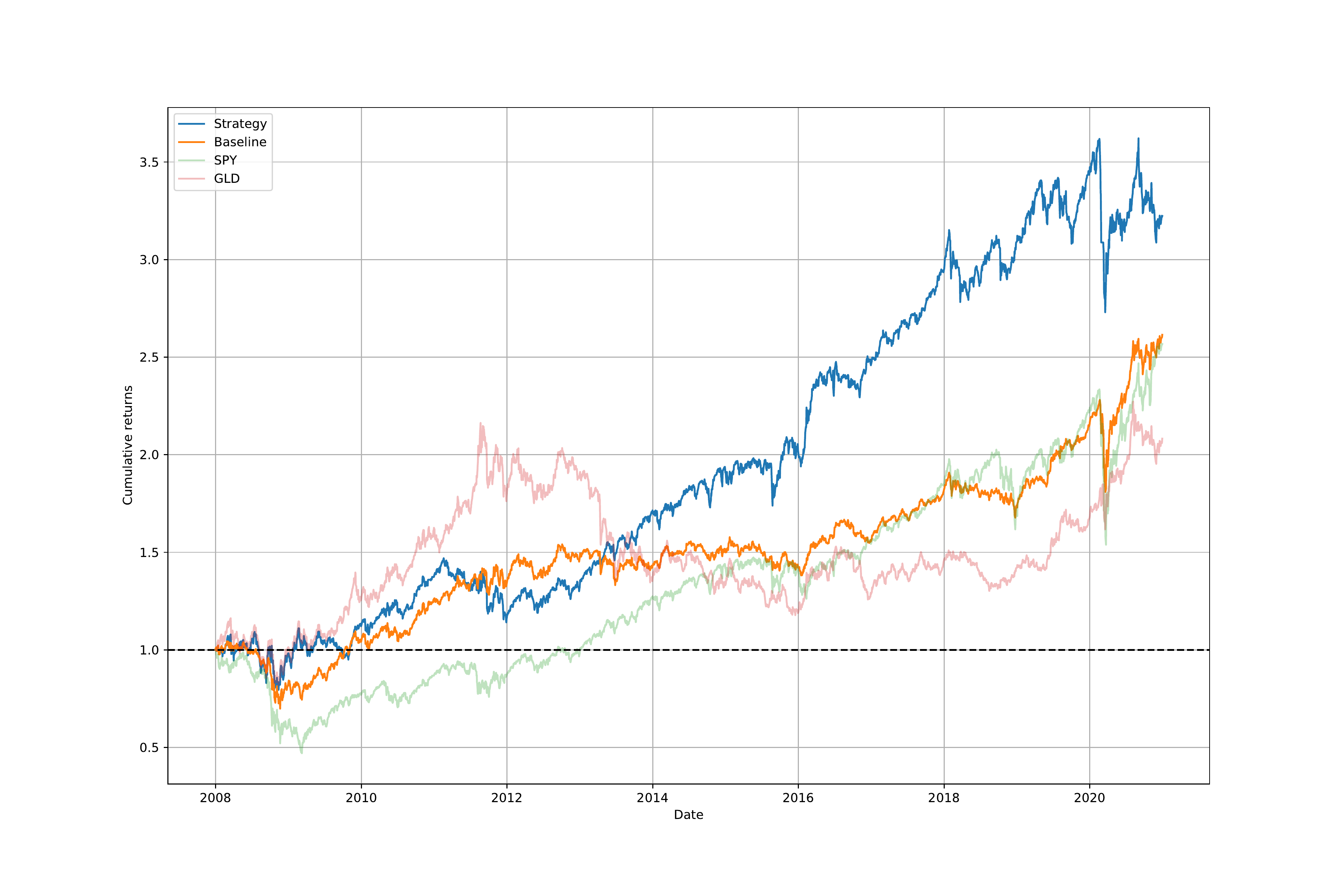}
\caption{Cumulative returns: January 2008 - December 2020}
\label{fig:2020bt}
\end{figure}

\begin{figure}[htpb]
\centering
\includegraphics[width=\textwidth]{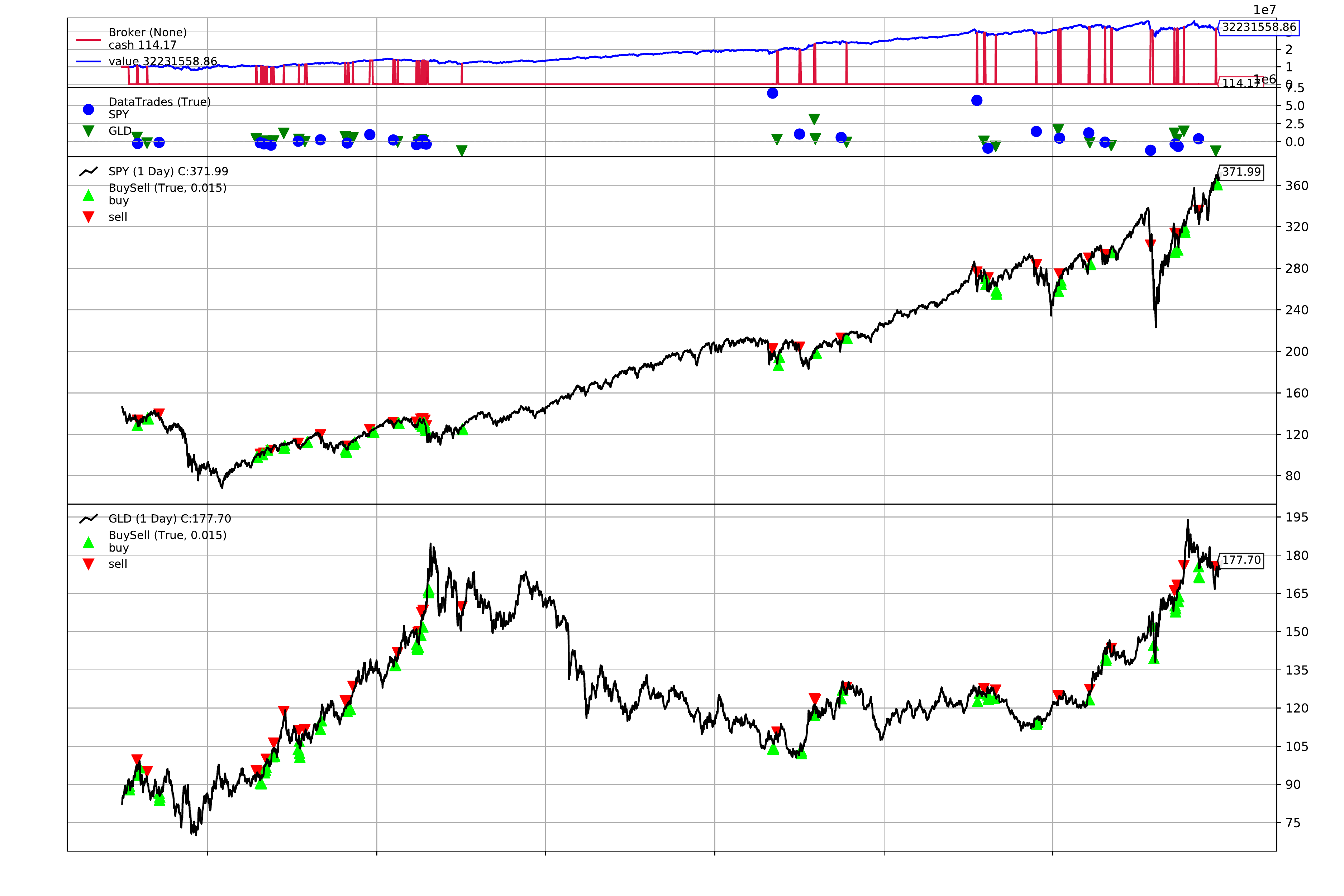}
\caption{Positions held by dynamic strategy: January 2008 - December 2020}
\label{fig:2020pos}
\end{figure}

Of the four strategies compared, our dynamic trading strategy has the best annualized returns, Sharpe ratio, Sortino ratio and Calmar ratio, and lowest drawdown. It has the second-lowest standard deviation of 0.10, close to the baseline static strategy's 0.0087. The most significant component of the Sharpe ratio's performance comes from strong annualized returns; the increased upside volatility is the main contributor to the standard deviation. Indeed, our strategy's Sortino ratio is about 33\% greater than that of the S\&P 500; this confirms that a significant degree of the penalty in the standard deviation and Sharpe ratio is generated from upside returns. That is, the strong annualized returns of our trading strategy are generated in a relatively volatile manner. This is unsurprising, given that the strategy generates performance due to market timing.

\subsection{Detailed analysis of performance over time}
In this section, we describe the performance in detail over various time periods, particularly during market crises. While we have reported our findings over one period 2008-2020, in fact, four separate learning and evaluation procedures have been performed. All four periods were successful for our strategy, visible in Figure \ref{fig:2020bt}.

First, the strategy performs well during the GFC. Our strategy generates the second-best returns during the GFC, surpassed only by gold. During the GFC, gold provided extraordinary returns for investors who 
invested prior to or during the crisis. After incurring a sharp drawdown, our strategy reallocates capital from S\&P 500 into gold and consequently outperforms equity markets until late 2011.

Next, the market experienced significant drawdown in December 2018. Given the brevity of this drawdown, our trading strategy is unable to reallocate capital away from the S\&P 500 into gold fast enough to meaningfully reduce the strategy's drawdown. After all, our strategy is predicated on identifying regimes, and allocating capital when new data are identified as similar to past phenomena. It reflects the delicate balance in the look back length $n$. If it were too long, trading decisions would be made too slowly; if it were too short, trading decisions would be made too frivolously.

The final significant market crisis during our window of analysis is the market turbulence associated with COVID-19. Our strategy performs well during this period. While it does experience a dip around March of 2020, it does allocate funds away from SPY quite early and avoids a much larger crash. The algorithm then benefits from the V-shaped recovery and reaches its previous peak before switching to GLD for the remainder of the year, leading to a small drawdown. The algorithm reverts back to SPY at the very last trade. We expect the algorithm to benefit from the market rally around the 2021 economic stimulus plan and the possibility of COVID herd immunity through vaccines.

During the four 4-year windows that make up the 2008-2020 experiment, the optimal look back length $n$ changes as follows. For the four windows, the optimized $n$ is 20, 25, 20, 27 and 20 for 2004-2008, 2008-2012, 2012-2016, 2016-2020, Jan 2020- Dec 2020 respectively. This suggests that continually updating the look back length is important, due to the dynamic nature of markets. The longest look back length is during 2012-2016, a bull market period with the greatest consistency and least volatility in the return profile. This suggests that regimes were more persistent and possibly easier to identify during the 2012-2016 period.

\section{Conclusion}
\label{conclusion}

This paper demonstrates an original means of clustering volatility regimes, highlighting it as a useful tool for descriptive analysis of financial time series and designing trading strategies. Results on both synthetic and real data are promising, with good validation scores across a range of synthetic data and significant simplification of real time series. The findings support previous work by \citet{Hamilton1989}, and \citet{Lamoureux1990} and many others who contributed to the idea of discrete changes in volatility regimes. Moreover, while previous models generally select the number of regimes in advance, our model applies self-tuning unsupervised learning to determine the number of clusters in its implementation. In real data, we showed that this is usually between 2 and 4, while remaining flexible enough to detect more during crises. Our method integrates well with others in the literature \citep{Guidolin2011,Ang2012,Campani2021}, as the determined number of volatility regimes can then be used in an alternative regime-switching model that requires this quantity to be set \textit{a priori}.

Additionally, the dynamic trading strategy performs well at avoiding periods of significant volatility and drawdown, and performs substantially better than the SPY in various market conditions. Our method continually updates its distributions and parameters, reflecting the need for ongoing learning of market conditions and volatility structure. The method is flexible and also integrates with other statistical and machine learning methods. For instance, one could replace the static safe haven of gold with a learned allocation of low beta assets as a dynamic safe haven.

The precise methodology and applications described in this paper are not an exhaustive representation of the utility of this method. As long as there is consistency between the regime characteristic of interest, the change point algorithm, and the distance metric between distributions, the method could easily be reworked for classification of regimes of alternative characteristics, and in other domains of study. Future work could make several substitutions or improvements to the methodology. One could apply change point tests that detect changes in the mean to identify clusters with positive or negative returns. Different clustering algorithms could be used to uncover novel structures: for example, DBSCAN \citep{ester1996density} could be used to find outlier distributions, or fuzzy clustering could be used to find distributions that might belong to multiple clusters. These additional insights could be useful for making new inferences about the time series. Currently, we recompute the distributions every four years; instead, we could make additional use of online change point detection to optimize this period in a more sophisticated manner. We could integrate more sophisticated procedures from metric geometry than simply minimizing the distance to the prior distributions. Finally, we could conceivably combine other machine learning methods of predicting volatility and online decision-making with our distance-based detection of high volatility historical regimes.

\section*{Acknowledgements}
Many thanks to Kerry Chen and Alex Judge for helpful edits and discussion.

\section*{Disclosure statement}
No potential conflict of interest was reported by the authors.

\section*{Funding}
No funding was received for this research.

\appendix
\section{Details of change point detection algorithm}
\label{cpa appendix}

In this section, we describe the \emph{change point detection} framework. Developed by \citet{Hawkins1977,Hawkins2005}, change point algorithms seek to determine breaks in a time series at which the stochastic properties of the underlying random variables change, and have become instrumental in time series analysis.

\subsection{General change point detection framework}
First, we outline the change point detection framework in general. A sequence of observations $x_1,x_2,...,x_n$ are drawn from random variables $X_1, X_2,...,X_n$. We wish to determine points $\tau_1,...,\tau_m$ at which the distributions change. One assumes that the underlying random variables are independent and identically distributed between change points. One can summarize this with the following notation, following \citet{Ross_cpm}:

\begin{equation*}
    X_{i} \sim 
    \begin{cases}
      F_{0} \text{ if } i \leq \tau_1 \\
      F_{1} \text{ if } \tau_1 < i \leq  \tau_2  \\
      F_{2} \text{ if } \tau_2 < i  \leq \tau_3,  \\
      \hdots
    \end{cases}
\end{equation*}

That is, one assumes $X_i$ is a random sampling of a different distribution over each time period $[\tau_{i},\tau_{i+1}]$. In order to meet the apparently restrictive assumption of independence of the data, one must usually perform an appropriate transformation of the data. The log quotient transformation, which yields the log returns from the closing price data, is one such transformation \citep{Gustafsson2000}.

\subsection{Rank of observations and Mood Test}
\label{Mood test}
\citet{Ross2013} points out the fact that log returns often exhibit heavy-tailed behaviour. As a result, a nonparametric test is needed to detect change points that do not \textit{a priori} assume the distribution of the data. The \emph{rank test} is one such test. Suppose there are two samples of observations from unknown distributions $A = \{r_{1,1} ,r_{1,2},...,r_{1,m} \}$ and $B = \{r_{2,1},r_{2,2},...,r_{2,n} \}$. Define the \emph{rank} of an observation $r \in A \cup B$ as follows:

\begin{align*}
\rank(r) &= \sum^{m}_{j}\mathbbm{1}_{(r \geq r_{1,j})} + \sum^{n}_{j}\mathbbm{1}_{(r \geq r_{2,j})} = \# \{ s \in A \cup B: r \geq s\} 
\end{align*}
A larger rank indicates a higher positioning in the ordering of the elements of $A$ and $B$. If both sets of samples have the same distribution, the median rank among $\{ \rank(r): r \in A \cup B \}$ is $\frac{1}{2}(n + m + 1)$. In this case, one would assume that both sets have a near equal split of the ranks.

The \emph{Mood test} determines the extent that each observation's rank differs from the median rank, thereby detecting differences in the distributions' variance. If the samples have different variances, then one set of samples would have more extreme values than the other, which means the ranks would not be even between the two sets. Specifically, the test statistic is as follows:
$$
M'_{m,n} = \sum_{i=1}^{m} (\rank(r_{1,i}) - (m+n+1)/2)^{2}
$$
This is appropriately normalized:
\begin{align*}
N &= m +n\\
\mu_{M'} &= \frac{1}{12}m(N^{2} - 1)\\
\sigma_{M'} &= \frac{1}{180}mn(N + 1)(N^{2} -4)\\
M_{mn} &= \frac{1}{\sigma_{M'}}(|M' - \mu_{M'}|)
\end{align*}
If $M_{mn}$ is greater than some threshold $h$, we reject the null hypothesis that the distributions have the same variance, and conclude they have different variances. As depicted in \ref{app4}, the log return time series are tail heavy but strongly mean and median centred. Thus, the Mood test reliably detects changes in the variance without being affected by changes in the median. Compare Sections 4 and 5 of \citet{mood1954} for this distinction.

\subsection{CPM algorithm} 
\label{cpm}
Ross' CPM algorithm \citep{Ross_cpm} works by feeding in one data point at a time. When a change point $\tau$ is detected, the algorithm restarts and proceeds from that point, so it suffices to describe how the algorithm determines its very first change point.

Suppose $x_1,...,x_N$ is a sequence for which no change point has been detected. For each $m=1,2,...,N$ define $n=N-m$, mirroring the notation of \ref{Mood test}, and compute the Mood test statistic $M_{m,n}$. If the maximum among these, $M_N=\max_{m+n=N} M_{m,n},$ exceeds a threshold parameter $h_N$, we declare a change point in the variance has occurred at $\hat{\tau}=\argmax_{m} M_{m,n}$. If the maximum such test statistic does not exceed the threshold parameter, feed in the next data point $x_{N+1}$ and continue. If a change point $\hat{\tau}=m$ is detected at time $N$, there has been a delay of $n$ units in its detection. This delay is necessary for the algorithm to examine data points on each side of the change point. The algorithm then restarts from the change point $\hat{\tau}$.

In our implementation of the algorithm, we read in at least 30 values before looking for another change point, so that all stationary periods have length at least 30. We choose our parameters $h$ in order to manage the number of false positives (Type I errors). Given an acceptability threshold $\alpha$, the following equations specify that this error should remain constant over time:

\begin{align*}
P(M_1 > h_1) &= \alpha \\
P(M_t > h_t | M_{t-1} \leq h_{t-1}, ...., M_1 \leq h_1) &= \alpha 
\end{align*}

In the event that no change point exists, a false positive will nonetheless be detected at time $1/\alpha$ on average. This quantity is the average run length parameter $\text{ARL}_0$ that is passed to CPM, which in term calculates the appropriate choice of $h_t$. In this case, $ARL_0$ is set to 10,000.

\section{Plots}
\label{app4}

We include detailed results for the experiments in Section \ref{results}. Figures \ref{fig:dow}, \ref{fig:n225} and \ref{fig:msci} depict the clustered distributions and volatility regimes of the Dow Jones, Nikkei and MSCI indices, respectively. Figures \ref{fig:aapl}, \ref{fig:amzn} and \ref{fig:brka} do so for individual firms Apple, Amazon and Berkshire Hathaway, respectively. Figures \ref{fig:xlf}, \ref{fig:ijs} and \ref{fig:ryt} depict popular ETFs XLF, IJS and RYT, respectively. Figures \ref{fig:aud}, \ref{fig:GBP} and \ref{fig:nzd} depict the distributions and regimes for the  AUD/USD, GBP/USD and NZD/AUD, respectively.
\vspace{1em}

\noindent We remark that all distribution plots are strongly centred in mean and median about zero. This is an important technical point for the Mood test to work correctly to detect changes in variance, as described in Appendix \ref{cpa appendix}.

\begin{figure}[htbp]
\centering
\subfigure[Kernel density estimation plots of Dow Jones distributions]{%
\resizebox*{14cm}{!}{\includegraphics{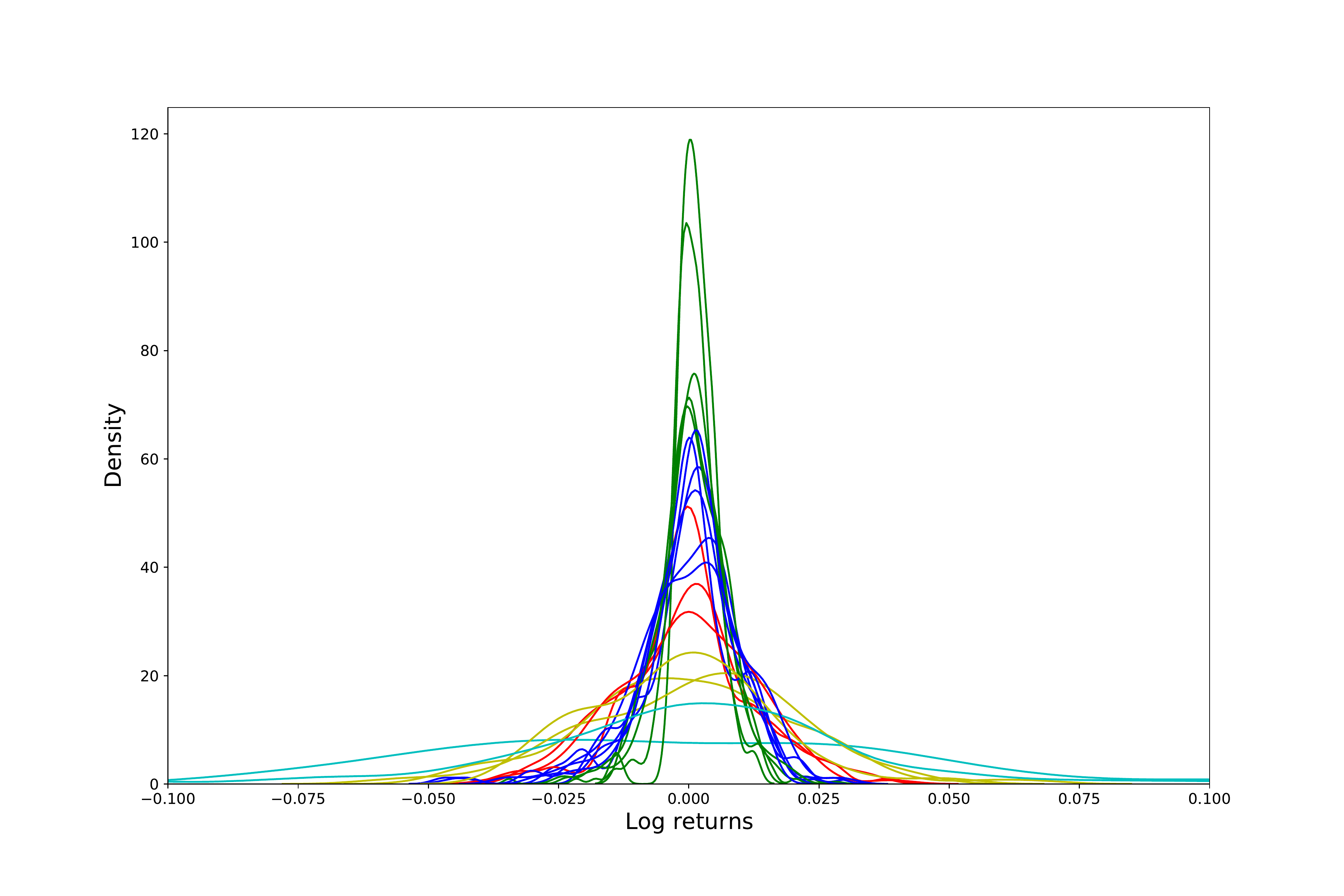}}}
\subfigure[Five determined volatility regimes for Dow Jones]{%
\resizebox*{14cm}{!}{\includegraphics{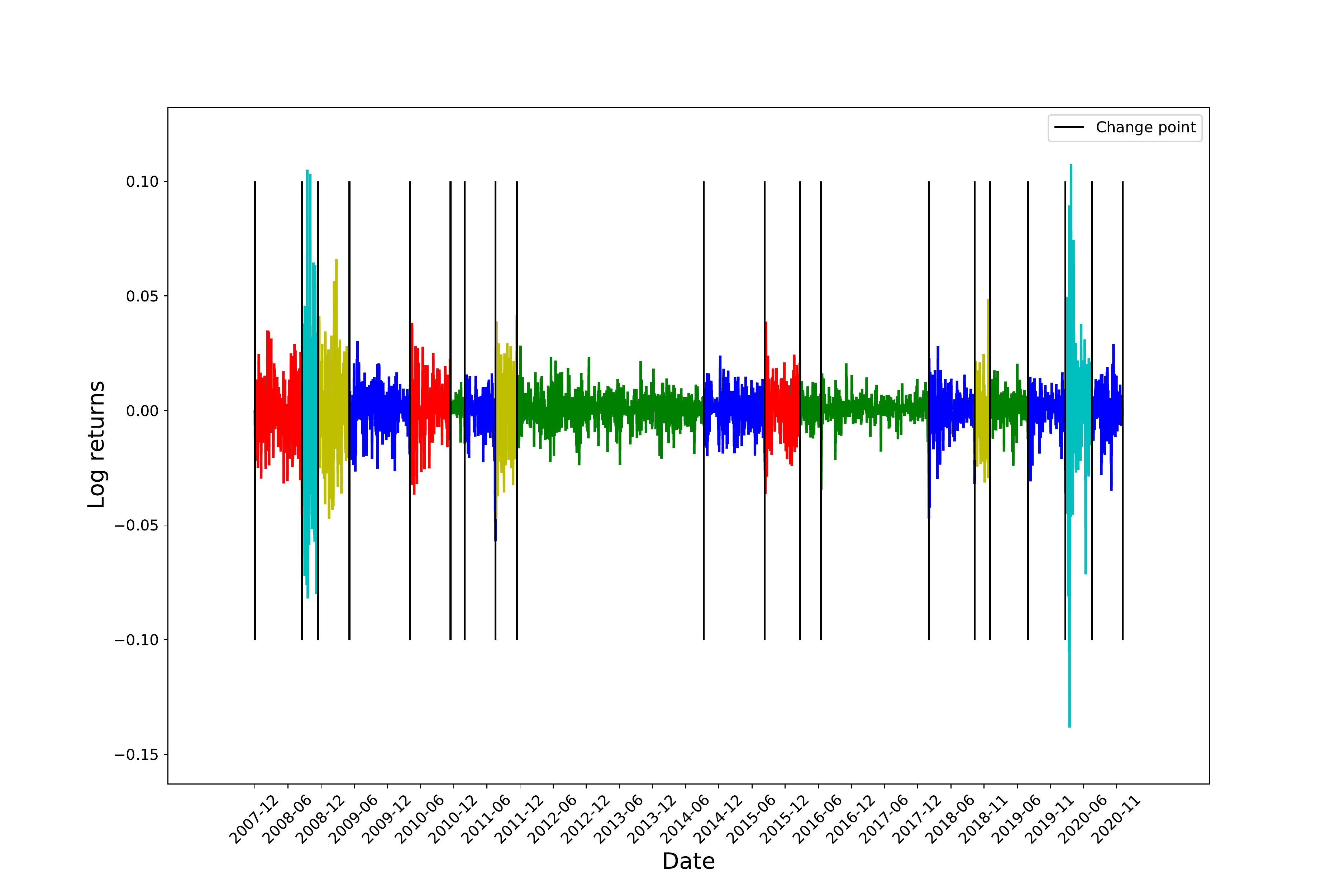}}}\hspace{5pt}
\caption{Volatility clustering results for Dow Jones} \label{fig:dow}
\end{figure}

\begin{figure}[htbp]
\centering
\subfigure[Kernel density estimation plots of Nikkei distributions]{%
\resizebox*{14cm}{!}{\includegraphics{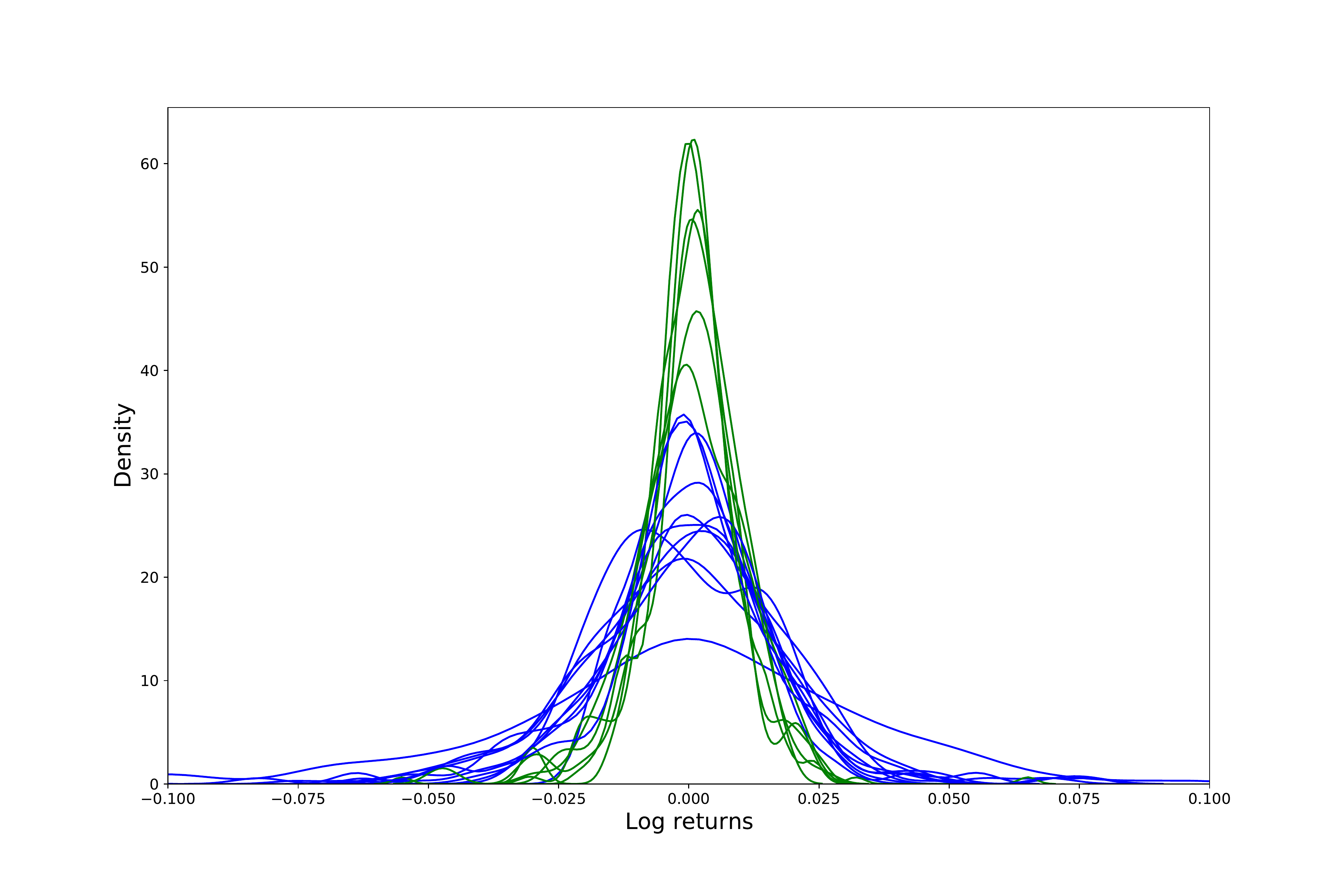}}}
\subfigure[Two determined volatility regimes for Nikkei]{%
\resizebox*{14cm}{!}{\includegraphics{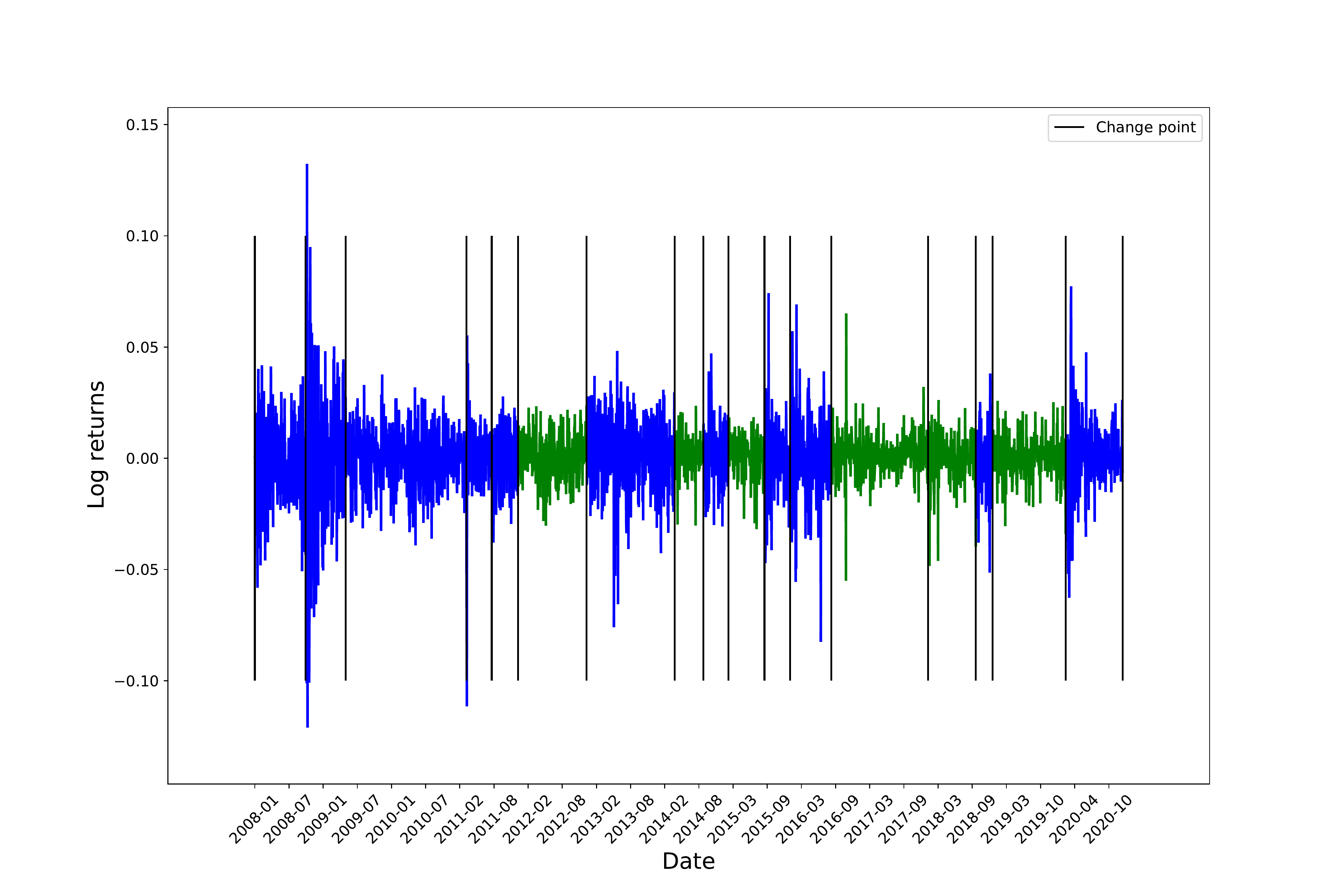}}}\hspace{5pt}
\caption{Volatility clustering results for Nikkei} \label{fig:n225}
\end{figure}

\begin{figure}[htbp]
\centering
\subfigure[Kernel density estimation plots of MSCI distributions]{%
\resizebox*{14cm}{!}{\includegraphics{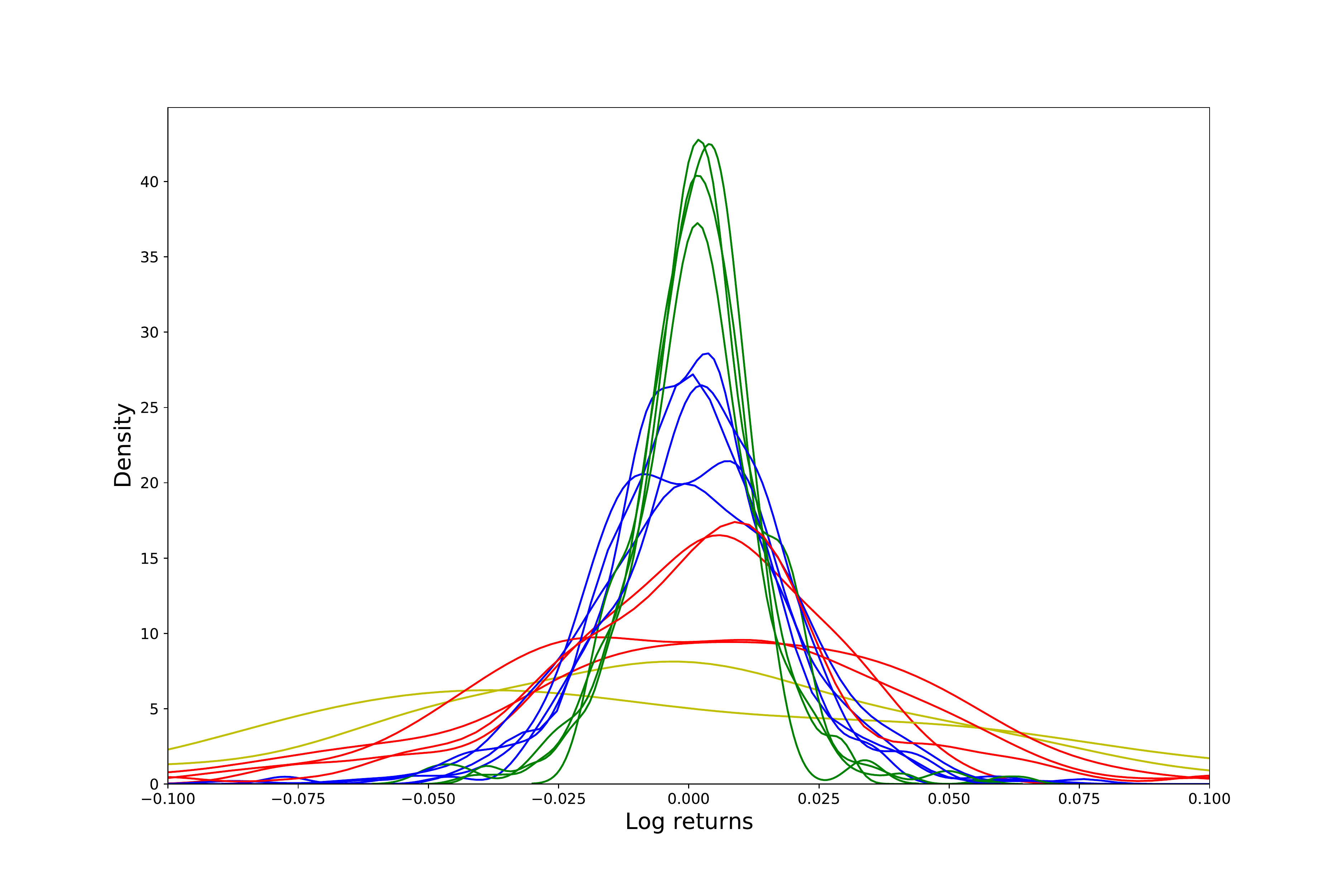}}}
\subfigure[Four determined volatility regimes for MSCI]{%
\resizebox*{14cm}{!}{\includegraphics{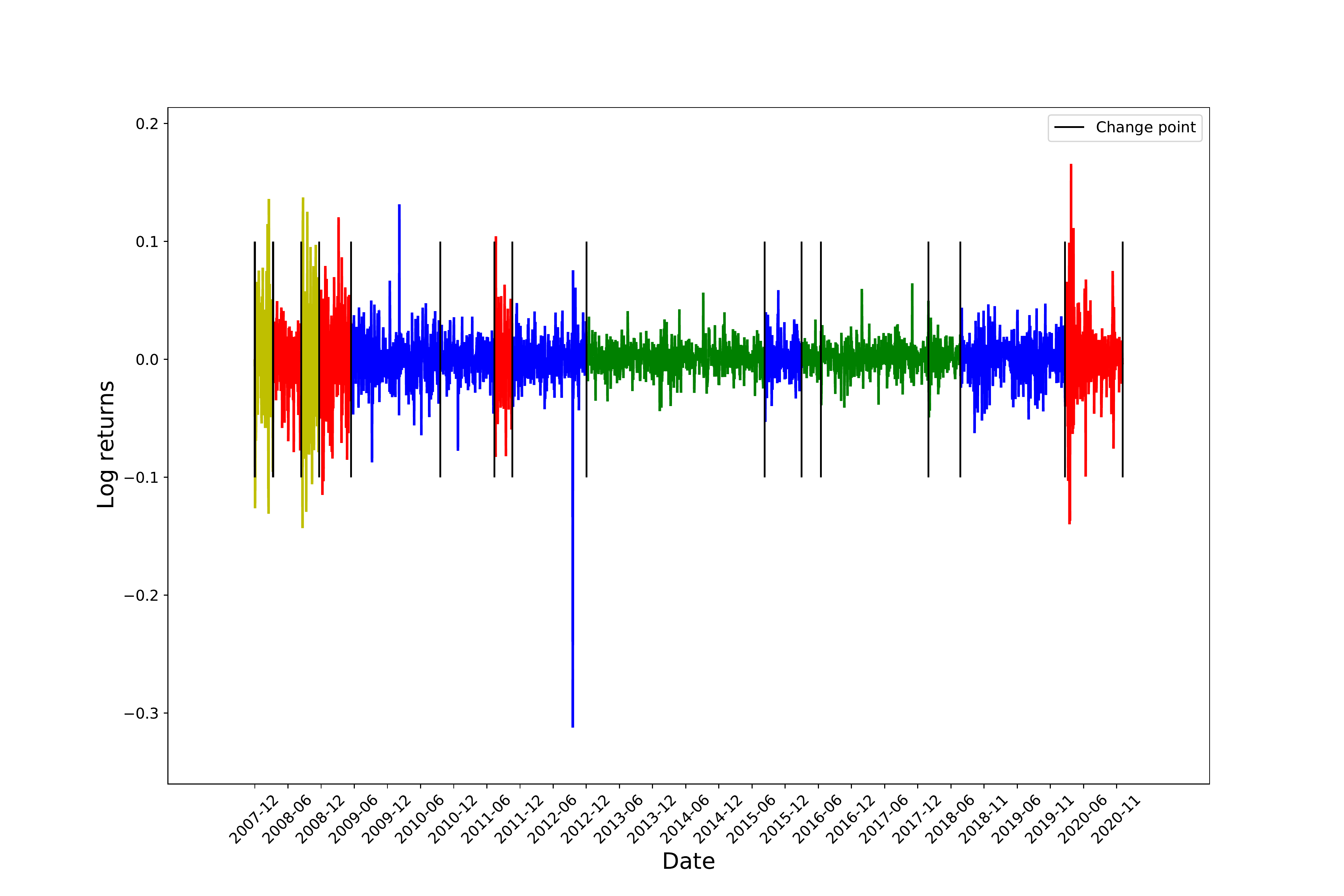}}}\hspace{5pt}
\caption{Volatility clustering results for MSCI} \label{fig:msci}
\end{figure}

\begin{figure}[htbp]
\centering
\subfigure[Kernel density estimation plots of AAPL distributions]{%
\resizebox*{14cm}{!}{\includegraphics{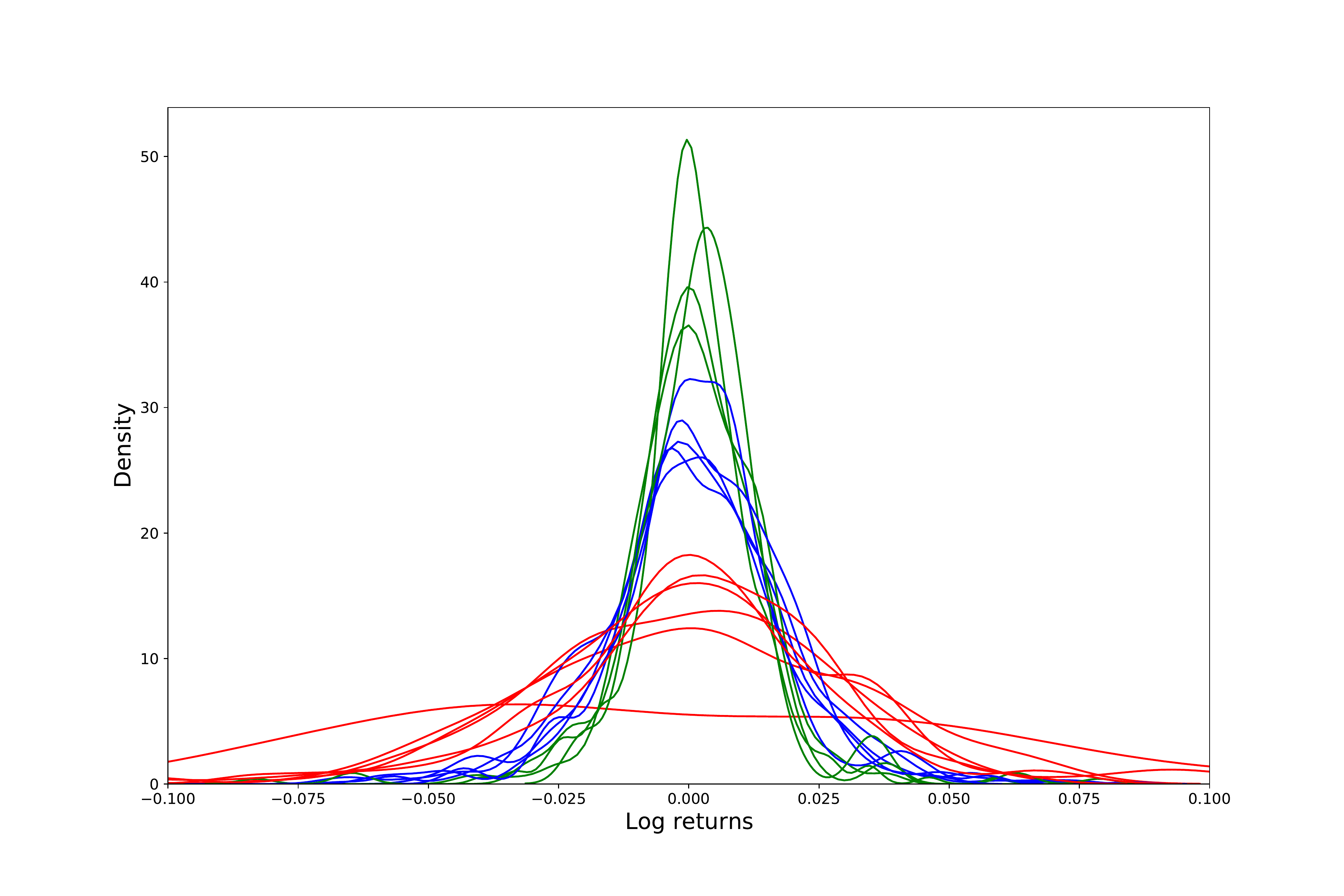}}}
\subfigure[Three determined volatility regimes for AAPL]{%
\resizebox*{14cm}{!}{\includegraphics{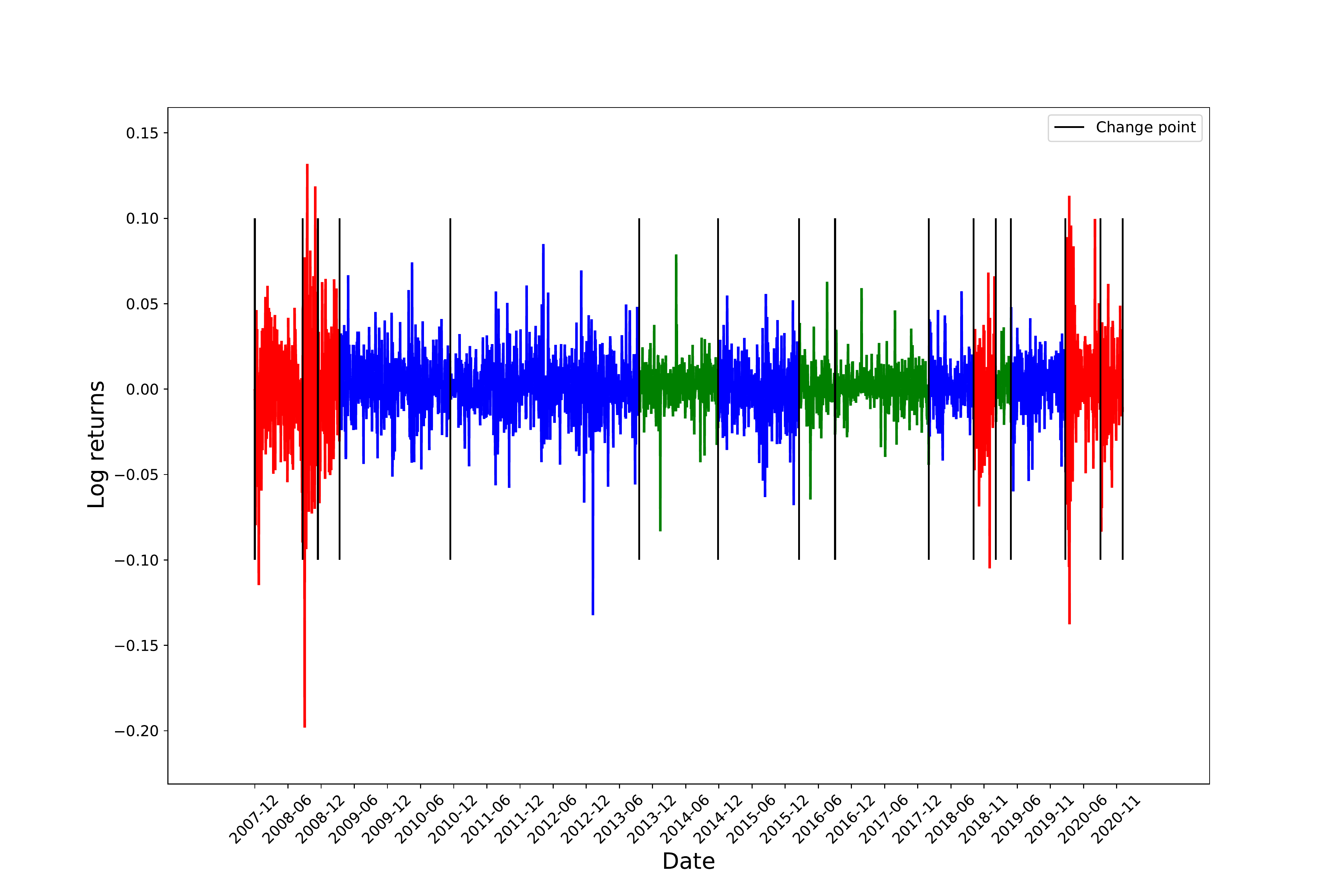}}}\hspace{5pt}
\caption{Volatility clustering results for AAPL} \label{fig:aapl}
\end{figure}

\begin{figure}[htbp]
\centering
\subfigure[Kernel density estimation plots of AMZN distributions]{%
\resizebox*{14cm}{!}{\includegraphics{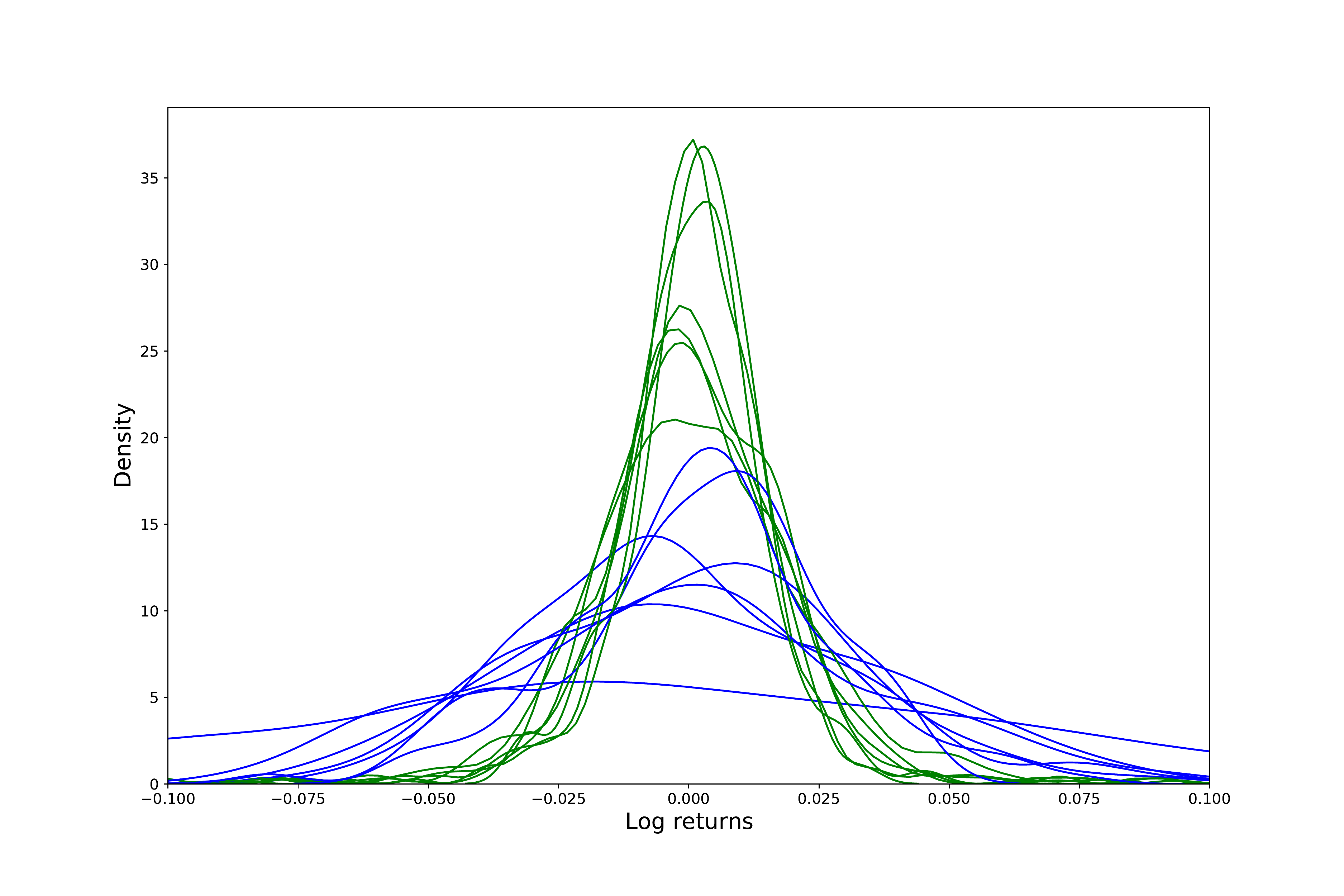}}}
\subfigure[Two determined volatility regimes for AMZN]{%
\resizebox*{14cm}{!}{\includegraphics{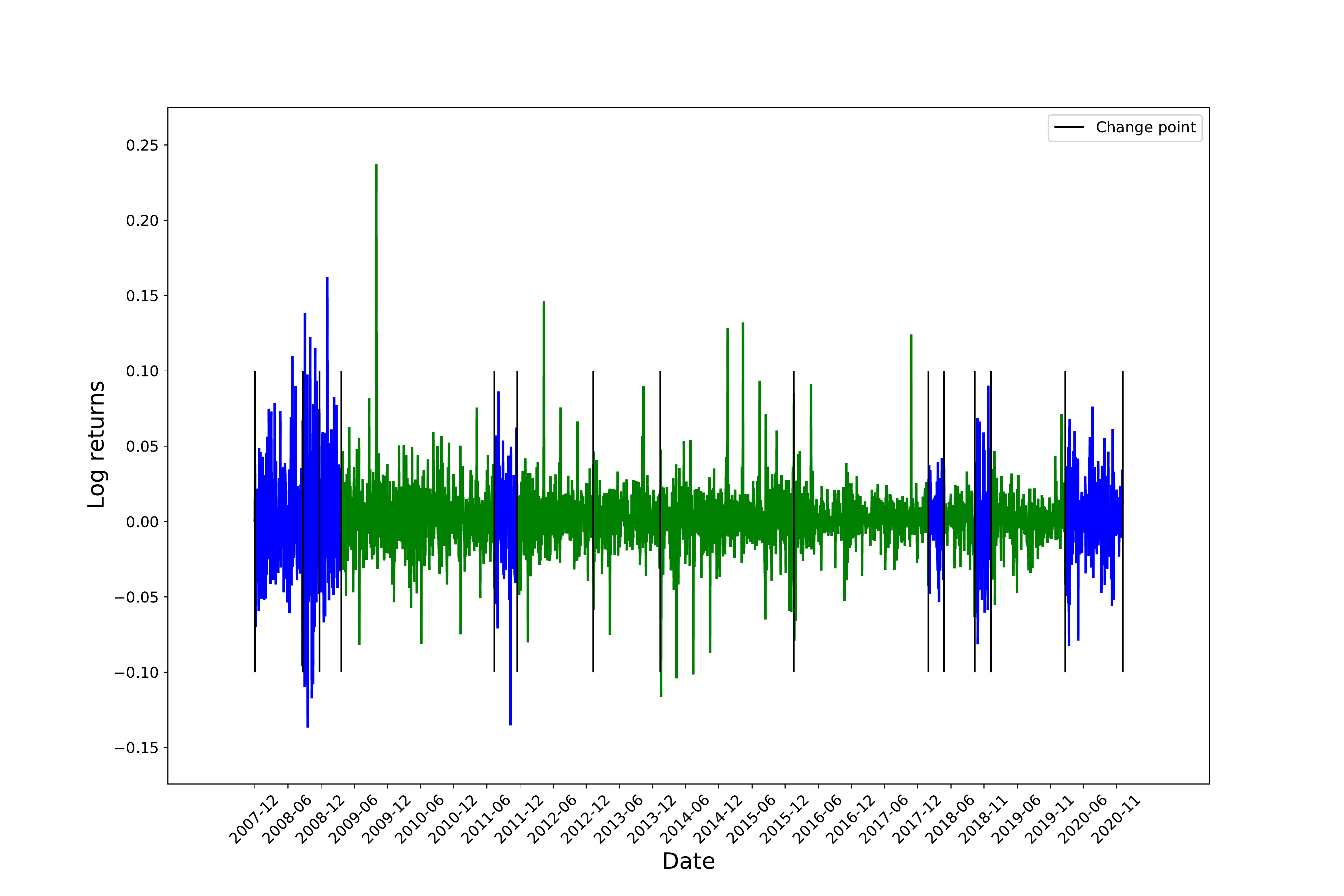}}}\hspace{5pt}
\caption{Volatility clustering results for AMZN} \label{fig:amzn}
\end{figure}

\begin{figure}[htbp]
\centering
\subfigure[Kernel density estimation plots of BRK-A distributions]{%
\resizebox*{14cm}{!}{\includegraphics{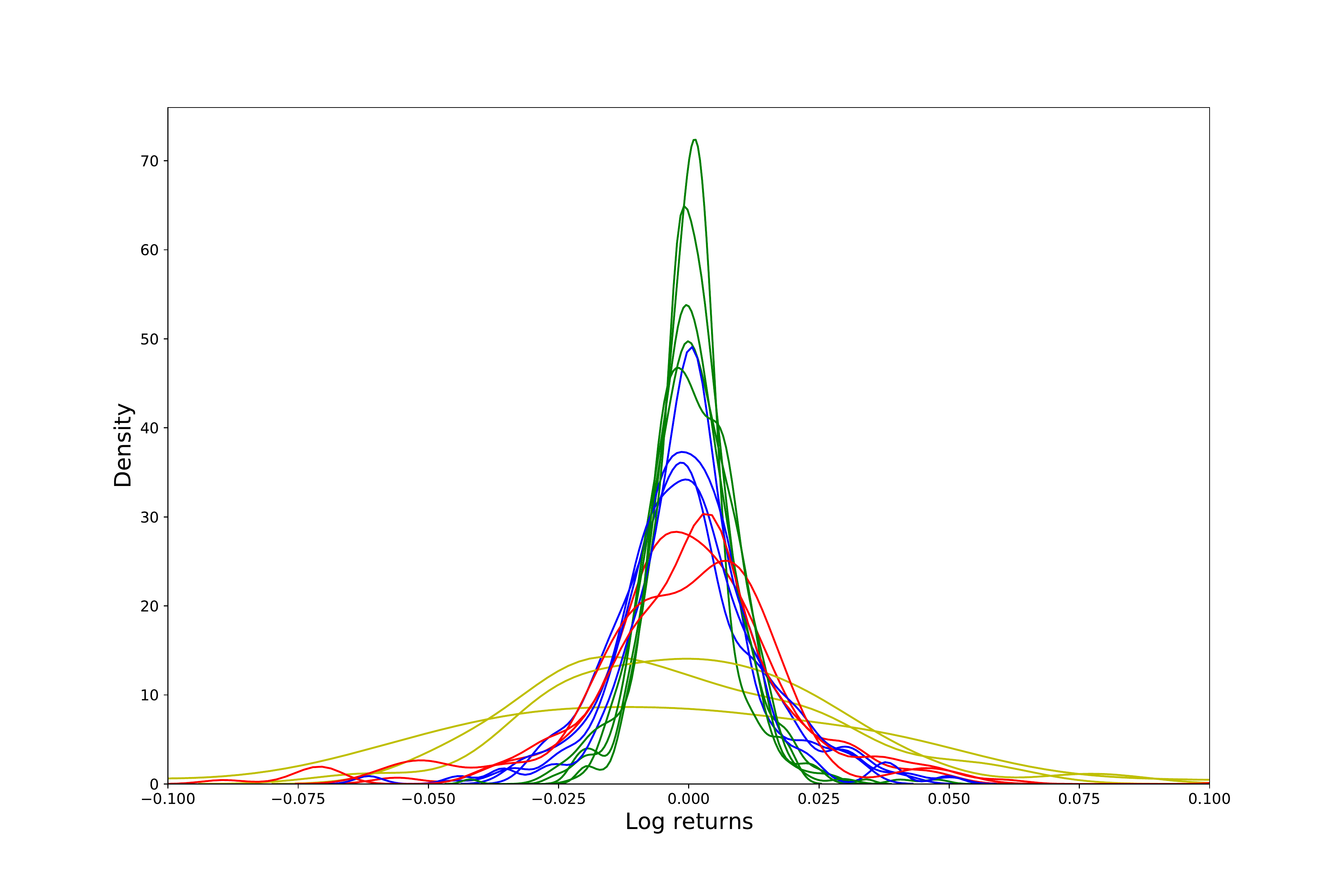}}}
\subfigure[Four determined volatility regimes for BRK-A]{%
\resizebox*{14cm}{!}{\includegraphics{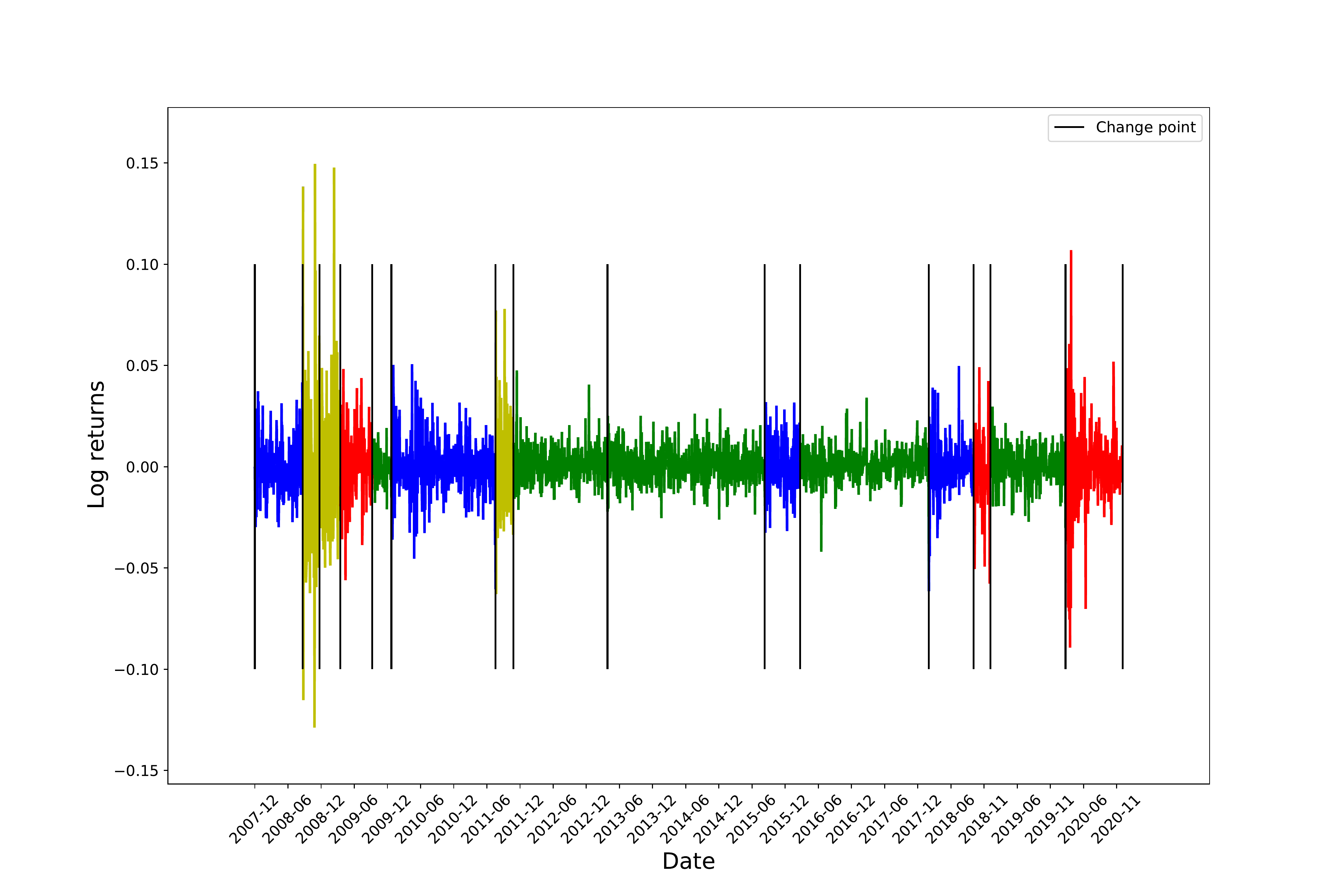}}}\hspace{5pt}
\caption{Volatility clustering results for BRK-A} \label{fig:brka}
\end{figure}

\begin{figure}[htbp]
\centering
\subfigure[Kernel density estimation plots of XLF distributions]{%
\resizebox*{14cm}{!}{\includegraphics{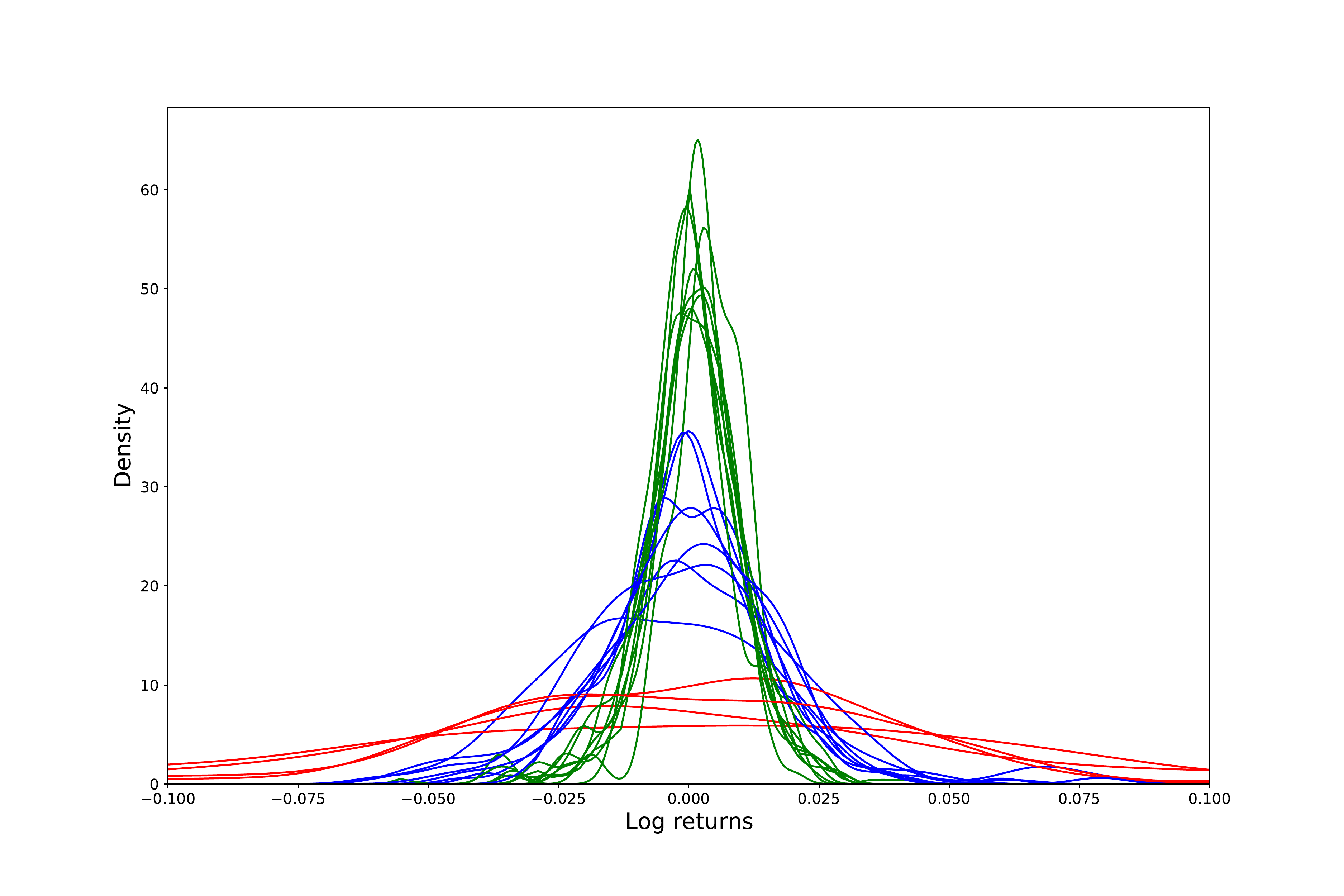}}}
\subfigure[Three determined volatility regimes for XLF]{%
\resizebox*{14cm}{!}{\includegraphics{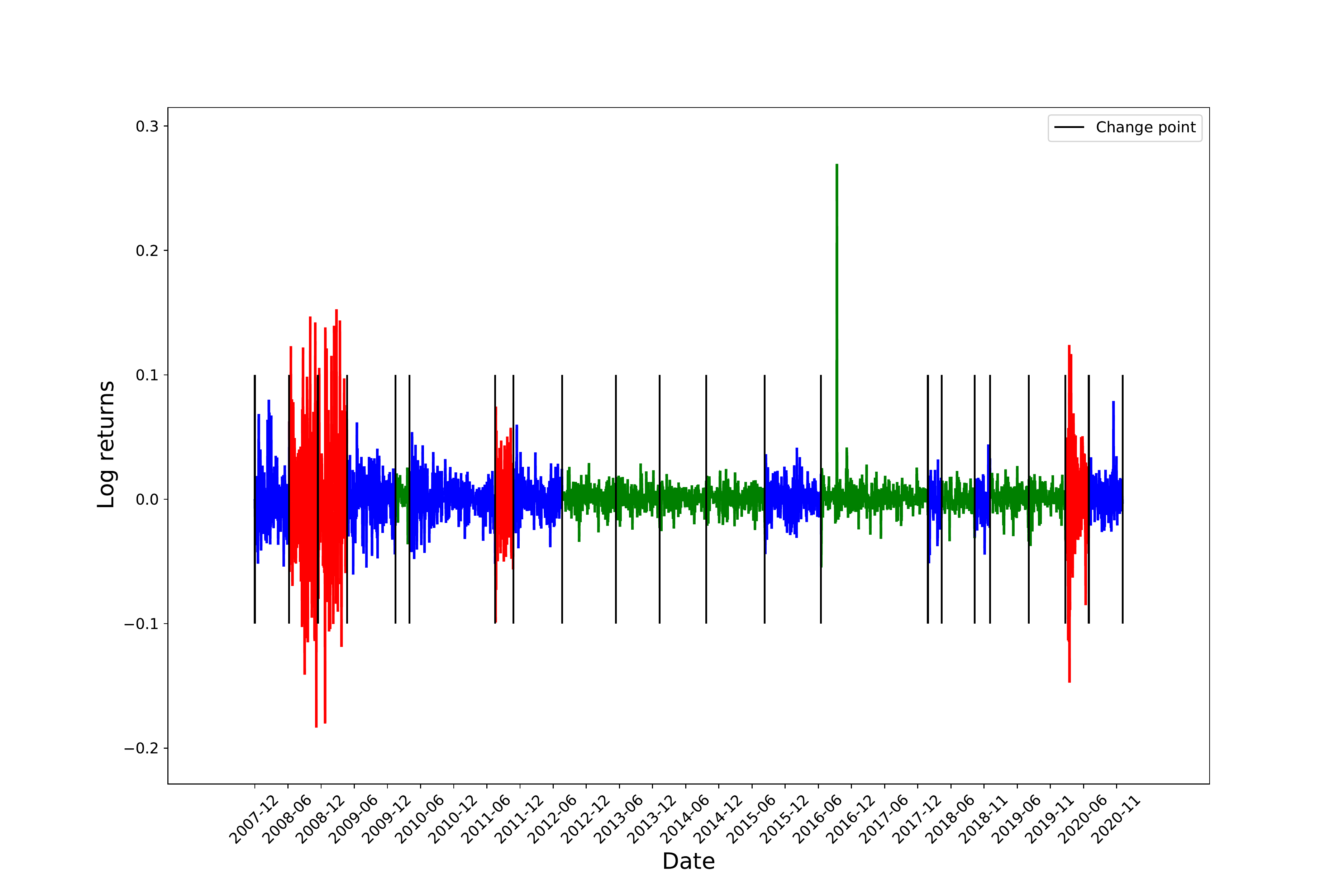}}}\hspace{5pt}
\caption{Volatility clustering results for XLF} \label{fig:xlf}
\end{figure}

\begin{figure}[htbp]
\centering
\subfigure[Kernel density estimation plots of IJS distributions]{%
\resizebox*{14cm}{!}{\includegraphics{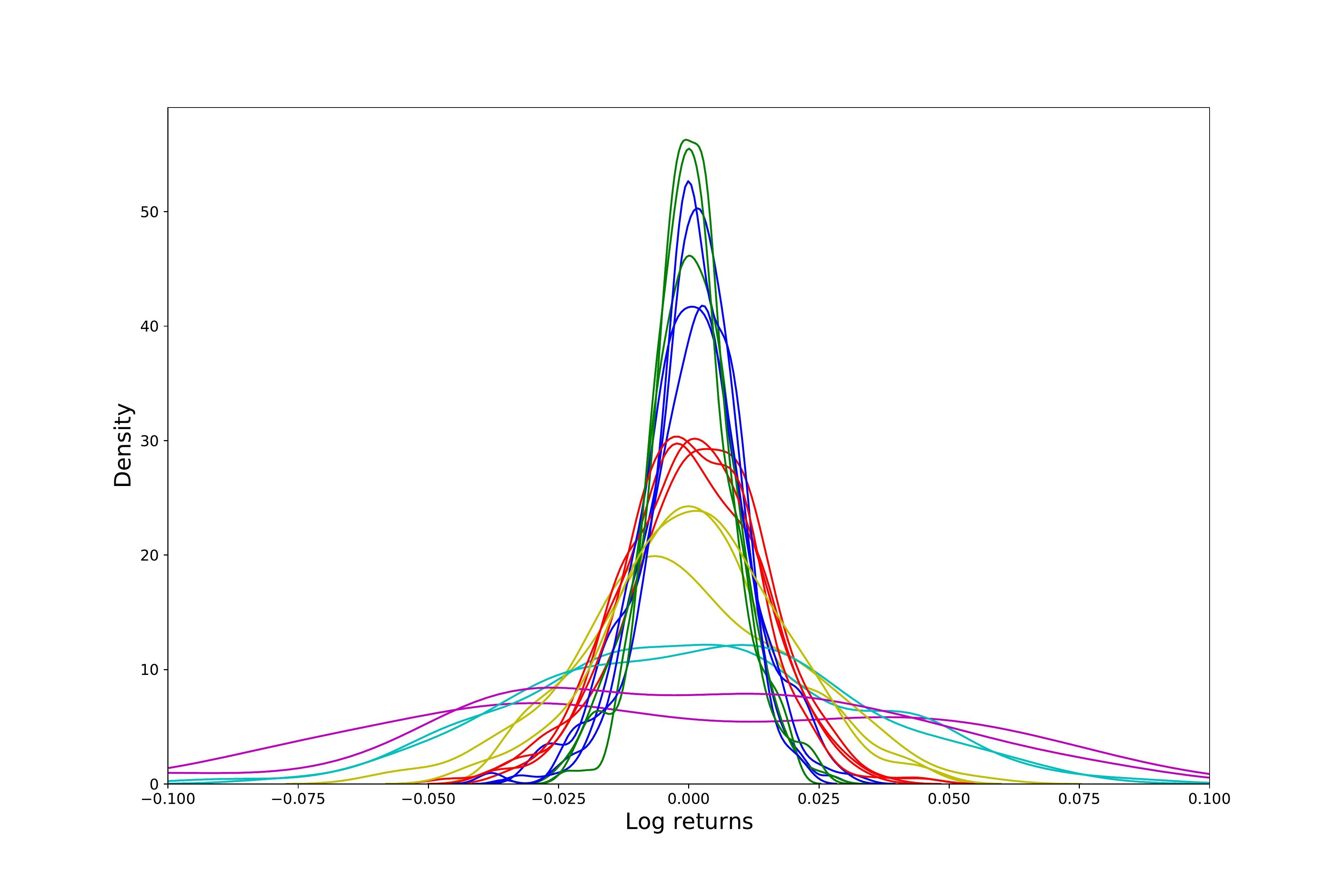}}}
\subfigure[Six determined volatility regimes for IJS]{%
\resizebox*{14cm}{!}{\includegraphics{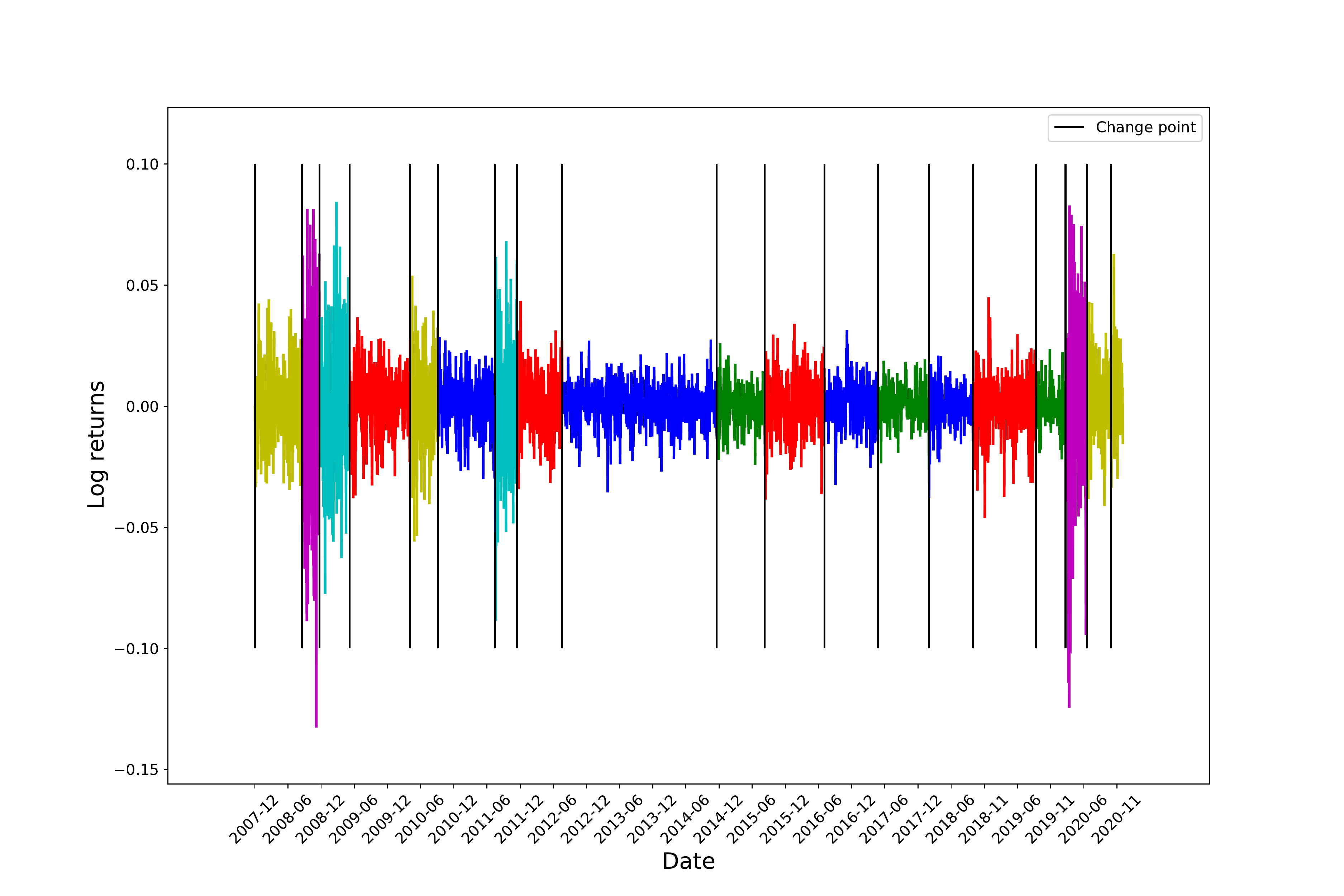}}}\hspace{5pt}
\caption{Volatility clustering results for IJS} \label{fig:ijs}
\end{figure}

\begin{figure}[htbp]
\centering
\subfigure[Kernel density estimation plots of RYT distributions]{%
\resizebox*{14cm}{!}{\includegraphics{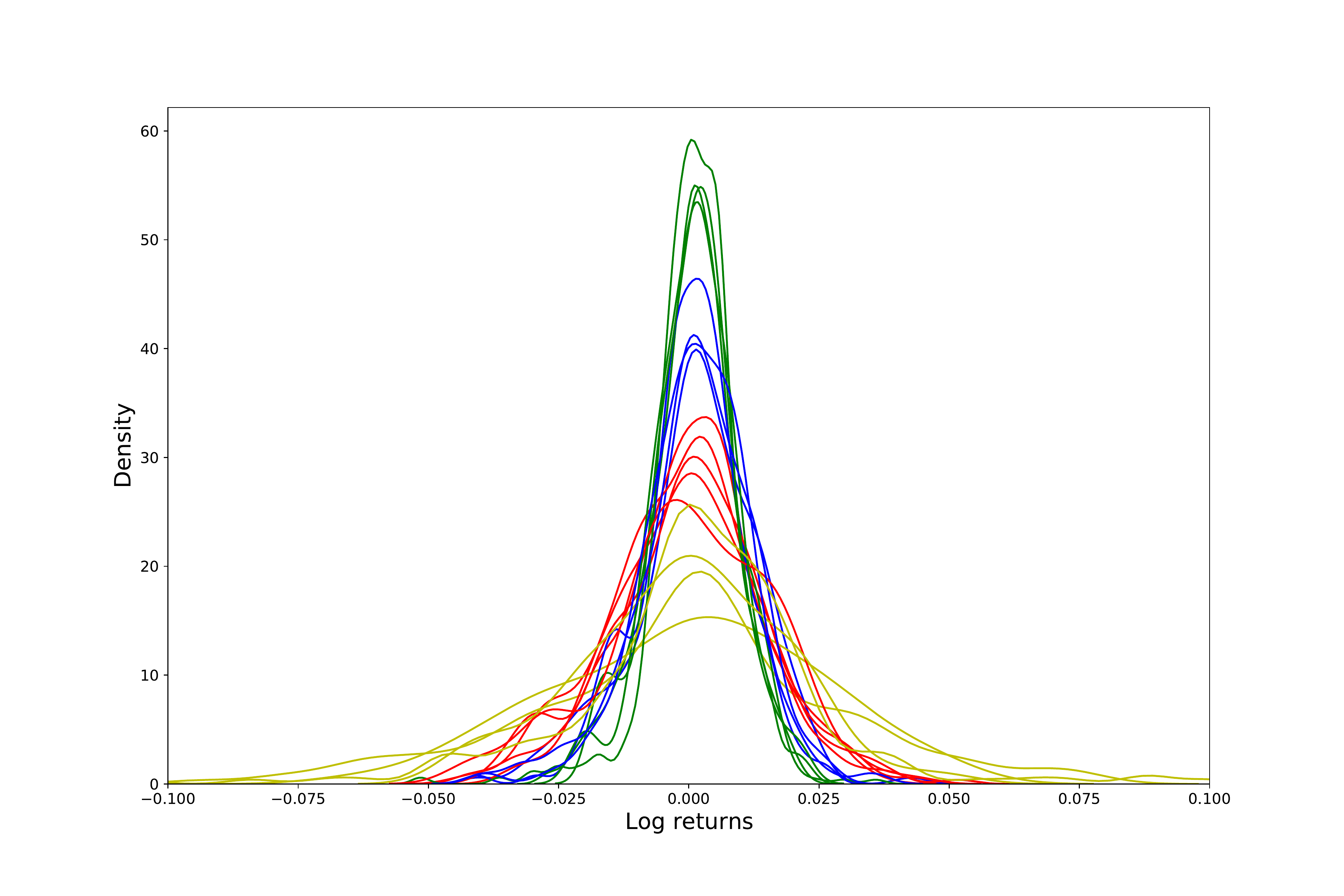}}}
\subfigure[Four determined volatility regimes for RYT]{%
\resizebox*{14cm}{!}{\includegraphics{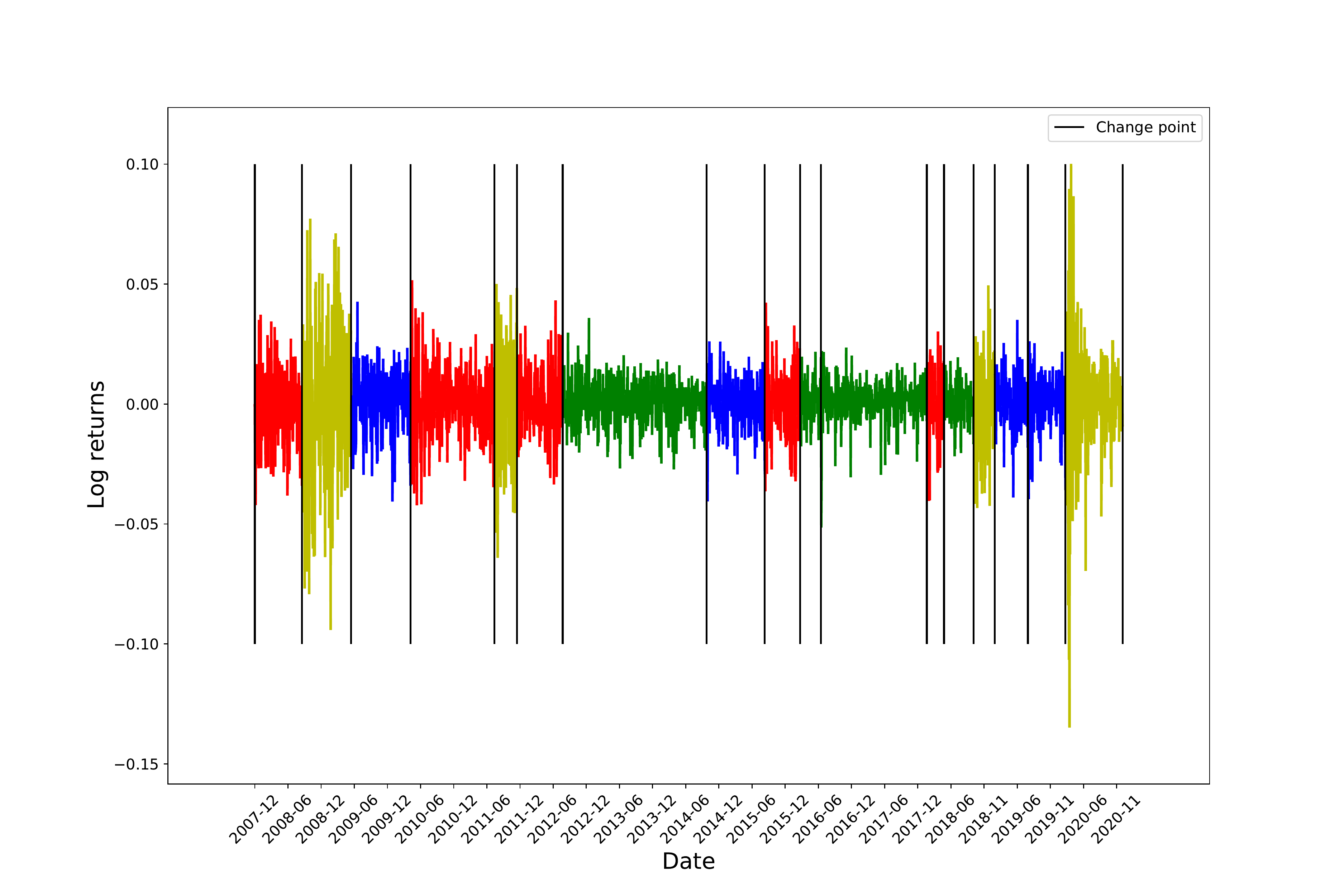}}}\hspace{5pt}
\caption{Volatility clustering results for RYT} \label{fig:ryt}
\end{figure}

\begin{figure}[htbp]
\centering
\subfigure[Kernel density estimation plots of AUD/USD distributions]{%
\resizebox*{14cm}{!}{\includegraphics{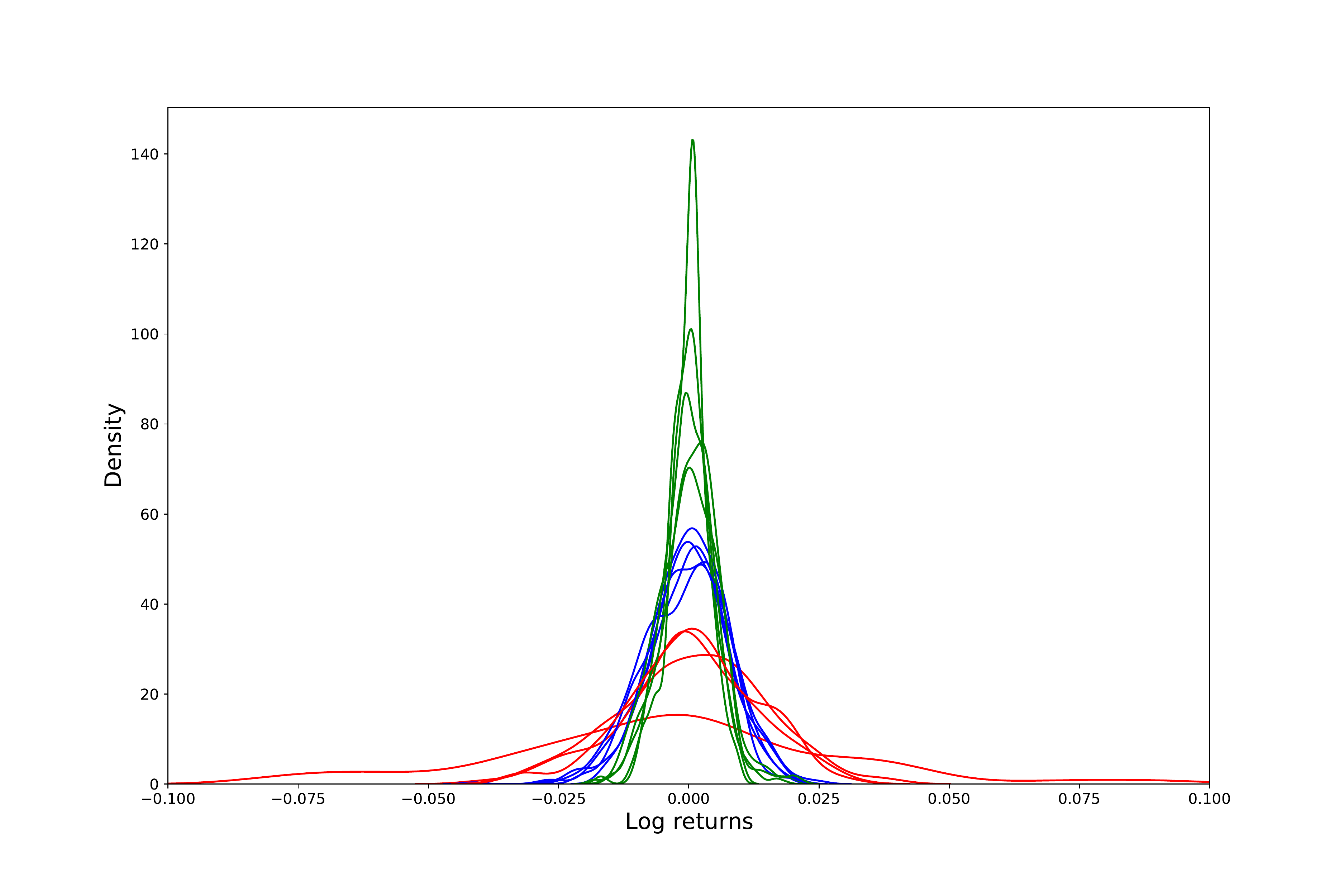}}}
\subfigure[Three determined volatility regimes for AUD/USD]{%
\resizebox*{14cm}{!}{\includegraphics{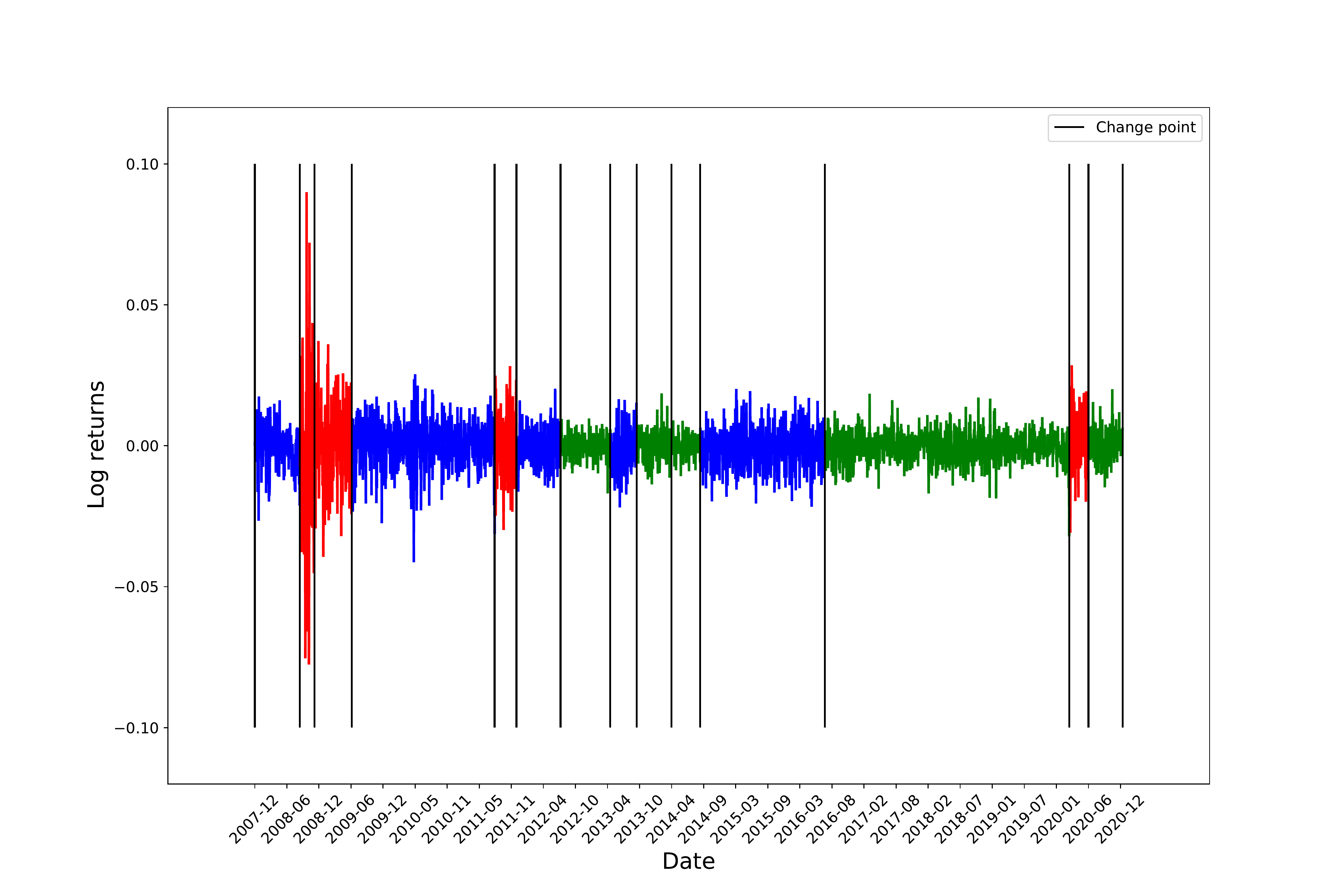}}}\hspace{5pt}
\caption{Volatility clustering results for AUD/USD} \label{fig:aud}
\end{figure}

\begin{figure}[htbp]
\centering
\subfigure[Kernel density estimation plots of GBP/USD distributions]{%
\resizebox*{14cm}{!}{\includegraphics{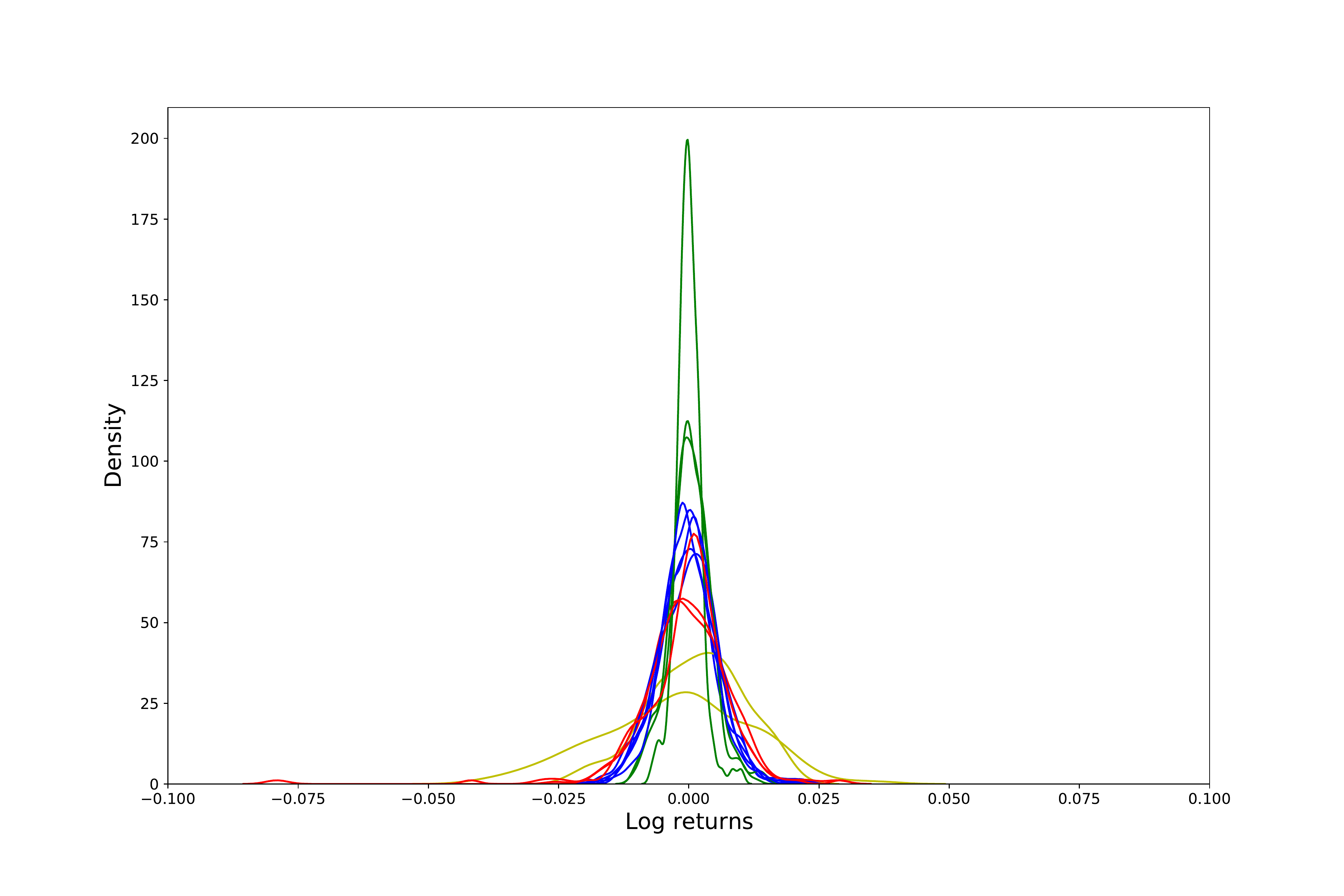}}}
\subfigure[Four determined volatility regimes for GBP/USD]{%
\resizebox*{14cm}{!}{\includegraphics{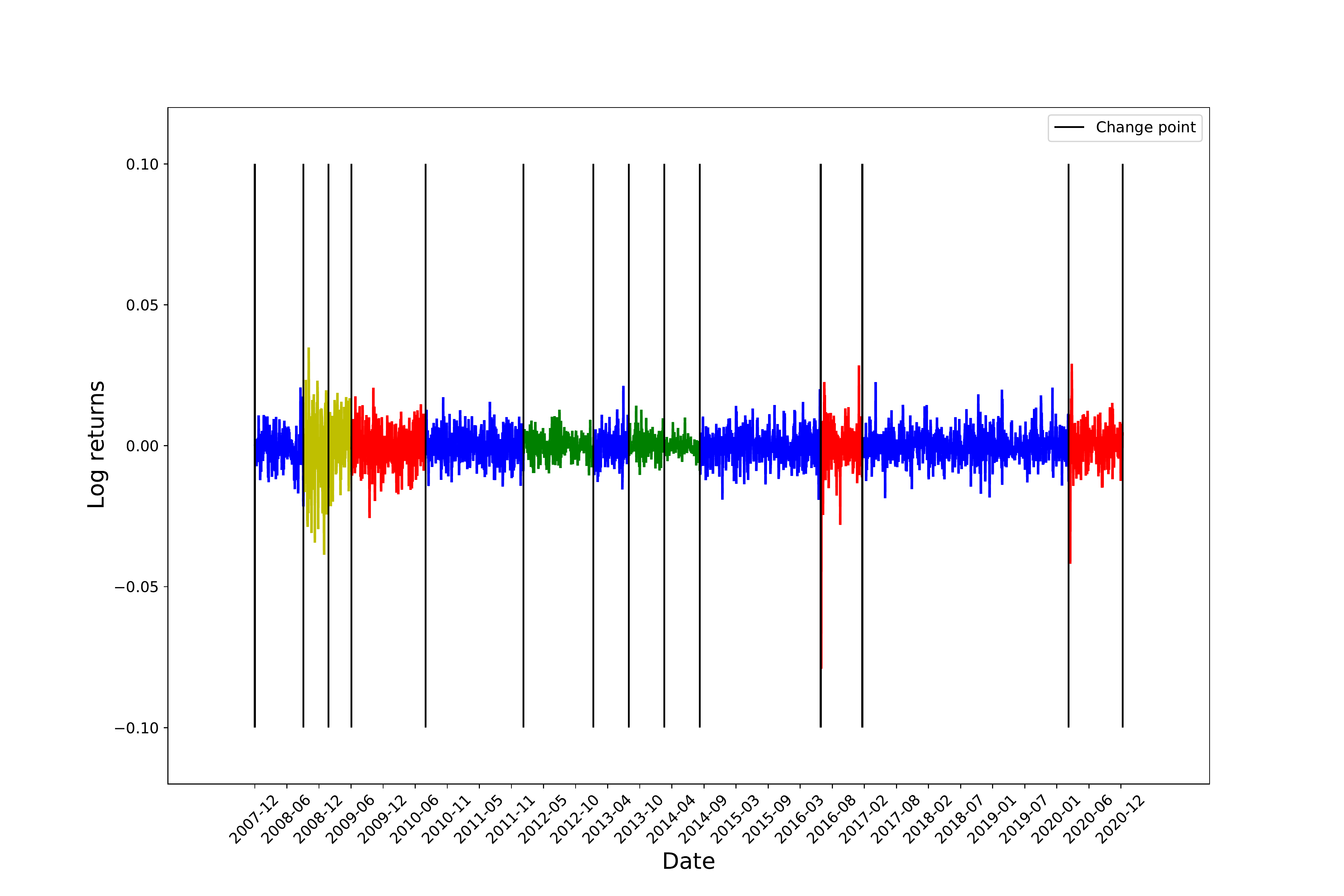}}}\hspace{5pt}
\caption{Volatility clustering results for GBP/USD} \label{fig:GBP}
\end{figure}

\begin{figure}[htbp]
\centering
\subfigure[Kernel density estimation plots of NZD/USD distributions]{%
\resizebox*{14cm}{!}{\includegraphics{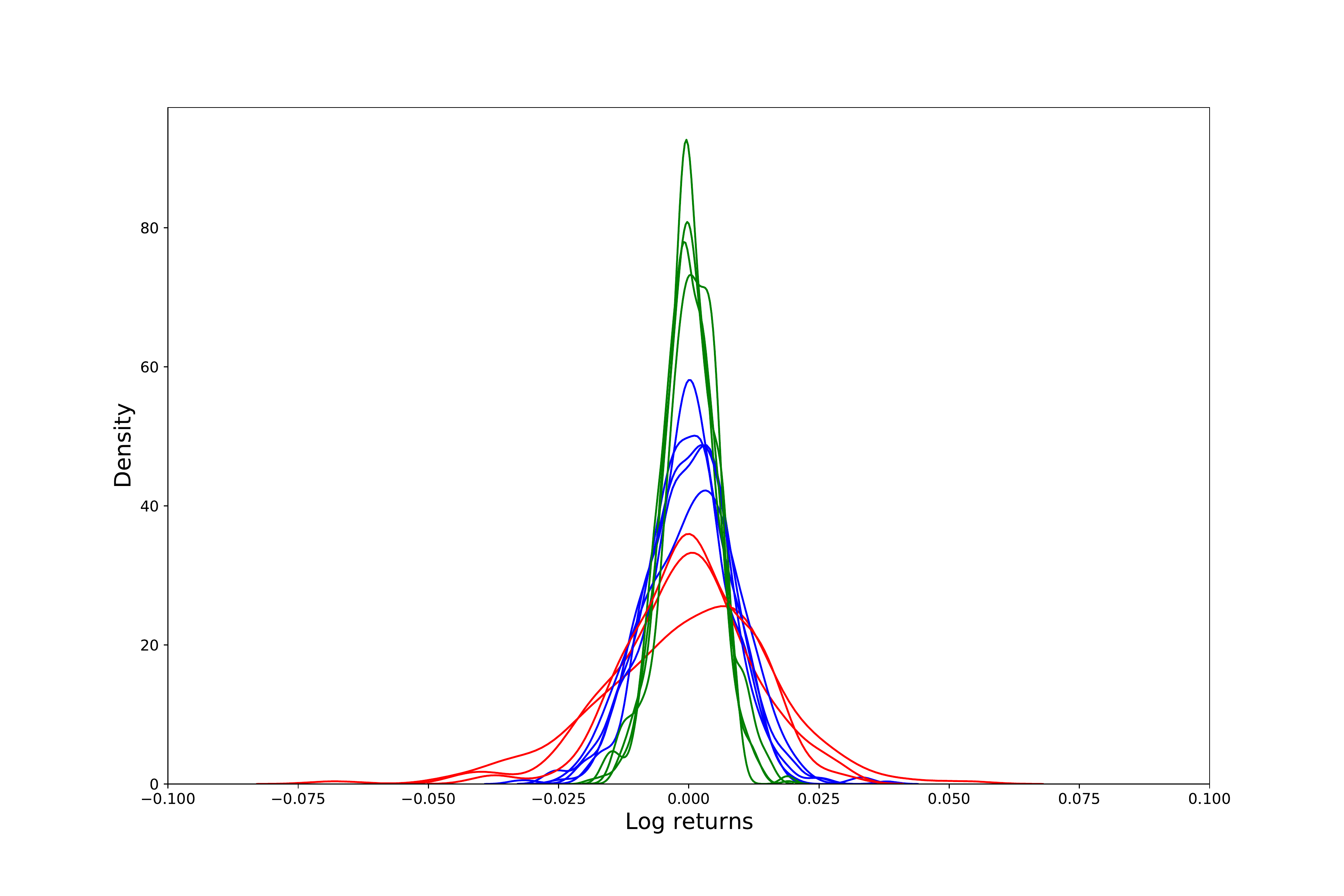}}}
\subfigure[Three determined volatility regimes for NZD/USD]{%
\resizebox*{14cm}{!}{\includegraphics{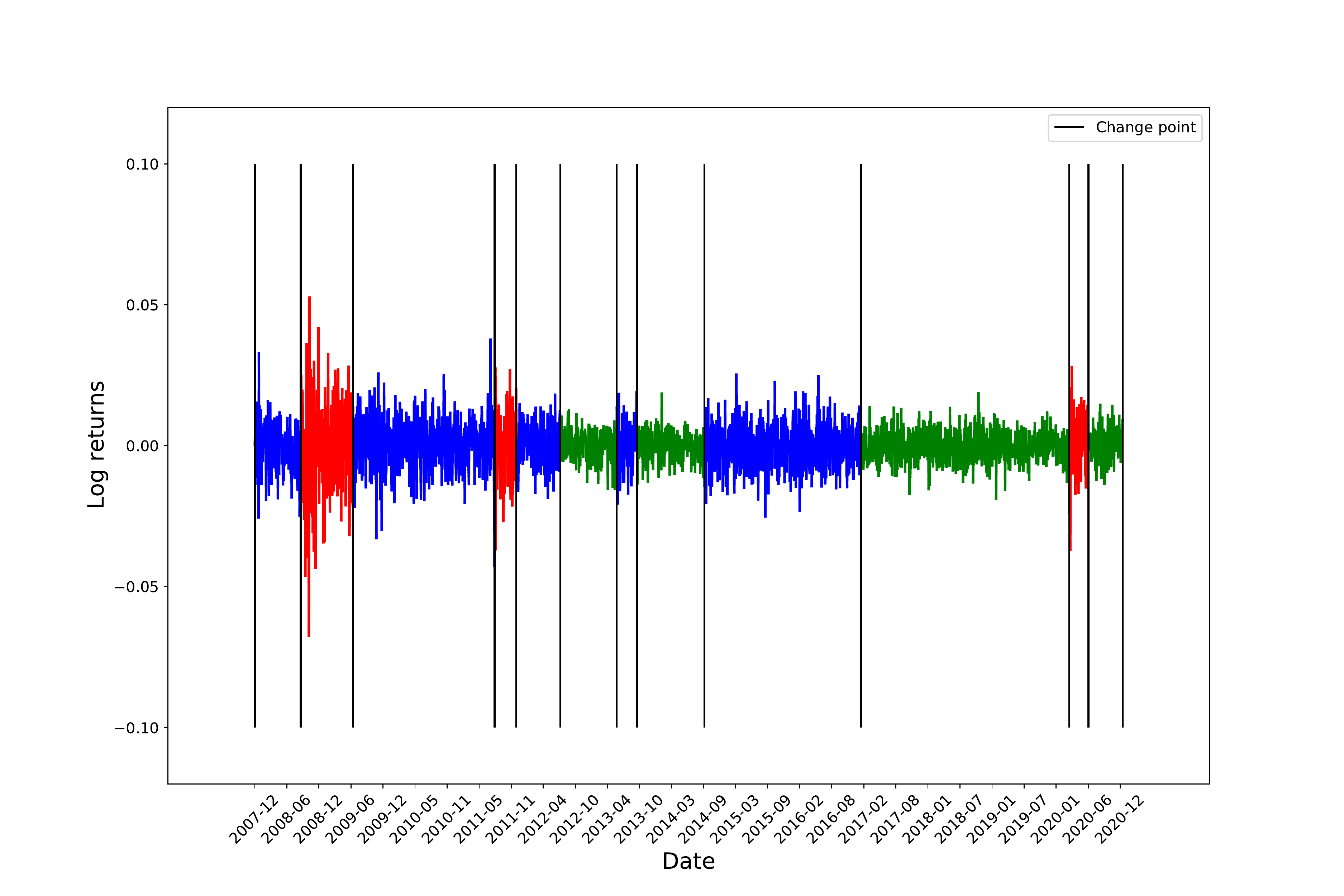}}}\hspace{5pt}
\caption{Volatility clustering results for NZD/USD} \label{fig:nzd}
\end{figure}

\clearpage

\bibliographystyle{tfcad}
\bibliography{_references}
\end{document}